\documentclass{article}
\usepackage[english]{babel}

\usepackage{amsmath,amsthm,amssymb}
\usepackage{mathtools,systeme}
\usepackage{setspace,dsfont,hyperref}
\usepackage{color}
\usepackage[super]{nth}
\usepackage{hyperref}       
\usepackage{url}            
\usepackage{booktabs}       
\usepackage{amsfonts}       
\usepackage{authblk}
\usepackage{microtype}      
\usepackage{xcolor}         
\usepackage{algorithm, algpseudocode}
\usepackage{cleveref}
\usepackage{tikz}
\usepackage{subcaption}
\usetikzlibrary{arrows.meta, positioning, decorations.pathreplacing}
\usepackage{pgfplots}
\usepackage[style=alphabetic, maxbibnames=99]{biblatex}
\pgfplotsset{compat=1.18}

\newtheorem{thm}{Theorem}[section]
\newtheorem{cor}[thm]{Corollary}
\newtheorem{lemma}[thm]{Lemma}
\newtheorem{prop}[thm]{Proposition}
\newtheorem{claim}[thm]{Claim}

\theoremstyle{definition}
\newtheorem{dfn}{Definition}[section]

\theoremstyle{definition}

\theoremstyle{definition}

\theoremstyle{remark}
\newtheorem{rem}{Remark}[section]

\newcommand{\eps}{\varepsilon}

\DeclarePairedDelimiter\norm{\lVert}{\rVert}

\usepackage{titlesec}
\titleformat{\section}{\normalfont\bfseries\large}{\thesection}{1em}{}
\titleformat{\subsection}{\normalfont\bfseries}{\thesubsection}{1em}{}

\begin{document}

\title{\bfseries Approximating CDTW Distance of Piecewise Algebraic Curves}
\author[1]{Alperen A. Erg\"ur}
\author[2]{Shamik Khowala}

\affil[1]{UT San Antonio, Mathematics and Computer Science Departments, San Antonio, TX}
\affil[2]{The Harker School, San Jose, CA}
\date{}

\maketitle

\begin{abstract}
    Curves as input data naturally arise in a variety of fields including finance, seismology, medicine, spatio-temporal data mining, malicious activity detection, and more. A common way to analyze these data sets is to do similarity matching or clustering. The most common metrics used for measuring similarity of curves are Dynamic Time Warping (DTW) and Fréchet distance. These metrics are sensitive to sampling rate and outliers respectively, and do not yield robust outcomes. Continuous Dynamic Time Warping (CDTW) is a more robust distance metric that improves upon DTW and Fréchet distances. Existing algorithms for CDTW are either exact algorithms that focus on non-Euclidean norms and piecewise linear curves, or approximation algorithms limited to piecewise linear curves. We present an approximation algorithm for computing the CDTW distance under Euclidean norm between piecewise (higher degree) algebraic curves. That is, we present a fully polynomial-time approximation scheme (FPTAS) of multiplicative error $\eps$, with  $O \left( (m+n)^{\frac{19}{6}} (\frac{1}{\varepsilon})^{\frac{10}{3}} \log \left( \frac{ (m+n) }{\varepsilon^2} \right) \right)$ complexity, where $m$ and $n$ are the number of pieces of the two input  curves.
\end{abstract}

\section{Introduction}

\noindent The ability to rigorously analyze time-series and trajectory data is essential in fields such as signature verification, animal tracking, and electrocardiography. Such analysis frequently involves clustering, classifying, and spotting trends. A common way to do this is by measuring the \textit{similarity} between curves. Many similarity measures have been considered, and the right choice often varies among applications.

Dynamic Time Warping (DTW) distance computes an alignment between two sequences accounting for the distances between all pairs of points, however, as an inherently discrete measure, it is sensitive to the sequences' sampling rates. Fréchet distance aligns curves accounting for the maximum distance between corresponding points; it is  a bottleneck measure  and inherently sensitive to outliers. Continuous Dynamic Time Warping (CDTW) distance overcomes both disadvantages by taking the best of both worlds: it combines the continuous nature of Fréchet distance with the summation-based nature of DTW distances.

The robustness of CDTW distance to outliers and sampling rates makes it a natural choice for a variety of applications. One of the first areas where CDTW was proposed is signature verification and handwriting analysis, where writing is naturally represented as continuous curves. Munich and Perona developed a CDTW-based algorithm for translation-invariant curve alignment that handled arbitrary point correspondence in signatures and  matching the continuous shape of the signature rather than isolated sample points \cite{munich1999continuous}. 


Clustering trajectories under CDTW avoids unwanted artifacts, such as zig-zags, that appear when using discrete DTW or the Fréchet distance. 
Brankovic, Buchin, Klaren, Nusser, Popov, and Wong studied $k$-median problem under CDTW distance, providing the first $(k,l)$-medians clustering algorithm using this metric. They utilize the squared Euclidean norm. The reasons for this is discussed in detail by Klaren \cite{brankovic2020k}. In short, Euclidean distance is the natural norm for most applications, and especially for clustering and regression tasks. Gudmundsson and Valladares gave a GPU-based approach for subtrajectory clustering using the continuous Fréchet distance \cite{gudmundsson2012gpu} (a variant of CDTW), which aids in the identification of movement patterns (e.g., flocks, leadership, and convergence) across large spatio-temporal datasets \cite{gudmundsson2008movement}. 


In the remaining parts of the introduction, we review algorithms and hardness results for DTW, Fréchet, and CDTW similarity metrics in detail. This gives context to state our results in \Cref{contributions}.

\subsection{Dynamic Time Warping}
\noindent DTW was largely popularized by Rabiner and Juang in the context of speech recognition, where it addressed differences in talking speed through time-normalization \cite{rabiner1993fundamentals}. Their approach relies on a weighted spectral distortion rather than a single fixed norm.

Given sequences $P = (p_1, \dots, p_n)$ and $Q = (q_1, \dots, q_m)$, DTW computes a monotone alignment $\pi$ that minimizes the cumulative cost
\[ DTW(P, Q) = \min_{\pi \in \mathcal{A}(n, m)} \sum_{(i, j) \in \pi} d(p_i,q_j), \]
where $\mathcal{A}(n, m)$ is the set of all valid alignments satisfying boundary, continuity, and monotonicity conditions and the distance function $d(p_i, q_j)$ is almost always the $L_1$ or $L_2$ norms.

The natural dynamic programming algorithm for computing DTW distance takes $O(mn)$ time. For temporal data, such as discrete time series in large databases, this quadratic complexity becomes too costly \cite{berndt1994using}. One workaround, FastDTW \cite{salvador2007toward}, approximates DTW in linear time and space. In another breakthrough paper, the quadratic barrier of computing DTW exactly was broken \cite{gold2018dynamic}. Gold and Sharir presented a deterministic algorithm which runs in $O(n^2/\log\log n)$ time and computes under any arbitrary fixed distance norm. For higher-dimensional sequences under the Euclidean $L_2$ norm, Agarwal, Fox, Pan, and Ying introduced $(1+\eps)-$approximation algorithms achieving near-linear time for $k$-packed or $k$-bounded curves as well as subquadratic time for backbone sequences \cite{agarwal2015approximating}.

DTW is not a metric since it violates the triangle inequality. For kernel methods, where a metric is required, Cuturi proposes Fast Global Alignment Kernels that instead sum over all possible alignments, not just the minimum \cite{cuturi2011fast}. Building on this, Random Warping Series (RWS) approximates using random features; it reduces computational complexity, enabling scaling to millions of time series  \cite{pmlr-v84-wu18b}.

\subsection{The Fréchet Distance}
\noindent The Fréchet distance was originally proposed by Maurice Fréchet in his 1906 doctoral thesis; Alt and Godau introduced an algorithm to calculate it in the early 1990s \cite{alt1995computing}. This measure provides a continuous alternative to DTW, taking the entire curve into account. An often cited analogy is that it represents the minimum leash length required for a man and his dog to traverse their respective curves without backtracking.

Originally, Alt and Godau proposed an $O(n^2\log n)$ algorithm which utilized Free Space Diagrams (FSDs). An FSD depicts regions where all pairs of points are within a certain distance $\eps$. However, Bringmann showed that no strongly subquadratic algorithms can compute the exact Fréchet distance between polygonal curves in any dimension, unless the Strong Exponential Time Hypothesis (SETH) fails \cite{bringmann2014walking}. Nevertheless, Buchin, Buchin, Meulemans, and Wolfgang achieved an improvement by precomputing small, repeating parts of the FSD and obtained an algorithm with  $O(n^2\sqrt{\log n}(\log \log n)^{3/2})$ complexity \cite{buchin2017four}.

Past work has also extended this concept to more complex geometric variants. For piecewise smooth curves, Rote showed the decision problem is solvable in $O(mn)$ time \cite{rote2007computing}, and he also extended this result to computing the Fréchet distance for piecewise algebraic curves (our focus). Conradi, Driemel, and Kolbe presented a simpler algorithm that achieves the same time complexity for piecewise smooth algebraic curves in arbitrary dimensions ($\mathbb{R}^d$) \cite{Conradi_Driemel_Kolbe_2025}. Our own results will refer to theirs when dealing with the algebraic curves. Buchin, Buchin, and Wang developed the first exact polynomial-time algorithm for partial curve matching --- maximizing the total length of subcurves within a distance $\delta$ of each other --- by simplifying it to a longest path problem solvable in $O(mn(m+n)\log(mn))$ time \cite{buchin2009exact}. For triangulated surfaces, computing the Fréchet distance remains a difficult problem, being only upper semi-computable \cite{buchin2007computability}. Meanwhile, Bringmann and Mulzer showed that a simple greedy algorithm for the discrete Fréchet distance gives a $2^{\Theta(n)}$-approximation in linear time \cite{bringmann2016approximability}.

\subsection{Continuous Dynamic Time Warping}
\noindent Continuous Dynamic Time Warping (CDTW), also called the integral Fréchet distance, addresses both DTW's sensitivity to sampling rates and the Fréchet distance's sensitivity to outliers; it computes a minimum-cost continuous alignment as in Fréchet distance but, instead of a maximum, it considers an integral, just as the sum of its discrete counterpart DTW.

Klaren explains a theoretical framework of CDTW, presenting a generalized definition that accommodates $L_1$, $L_2$, or $L_\infty$ norms on both the height function and the warping-path metric, including powers of a norm such as the squared Euclidean distance $L_2^2$ that we adopt here \cite{klaren2020continuous}. Klaren also presents three algorithms: ApproxCDTW for additive approximations, FastCDTW for faster execution, and ExactCDTW for exact computation of CDTW respectively.


Maheshwari, Sack, and Scheffer study two-dimensional polygonal curves using the Euclidean norm for the height function: They give a  pseudo-polynomial time $(1+\eps)$-approximation algorithm which runs in $O(\zeta^4 n^4 / \eps^2)$ time, where $\zeta$ is the maximal ratio of any pair of segment lengths from the input curves \cite{maheshwari2018approximating}.

Buchin, Nusser, and Wong presented the first exact algorithm for computing CDTW of one-dimensional curves, under the $L_1$ height  in $O(n^5)$ time by propagating continuous functions through a dynamic programming setup \cite{buchin2022computing}. Later Buchin, Buchin, Swiadek, and Wang showed that CDTW distance cannot be computed exactly under the Euclidean $L_2$ height norm using only algebraic operations, because the resulting integrands may involve transcendental numbers \cite{buchin2026constantfactorapproximationcontinuousdynamic}. This result motivates our own approximation algorithm in order to enable polynomial-time computation. In the same paper, they give the first polynomial-time constant-factor approximation for this problem. They achieve a 5-approximation in $O(n^5)$ time under the $L_1$ norm and a $(5+\eps)$-approximation in $O(n^5 / \eps^{1/2})$ time for any fixed polygonal norm \cite{buchin2026constantfactorapproximationcontinuousdynamic}.


In conclusion, developing more and more robust similarity measures is significant to support new applications and accurate analyses. CDTW is the most resilient of the measures reviewed previously, as it addresses both temporal distortions and outliers. The difficulty of computing CDTW, particularly for algebraic curves, is what we address. 


\subsection{Our Contribution} \label{contributions}
No progress has thus far been made toward exactly or approximately computing the CDTW distance between higher degree algebraic curves. Hardness results in \cite{buchin2026constantfactorapproximationcontinuousdynamic} effectively rule out exact computation for Euclidean height norms, so approximation is the natural target. In this paper, we present the first FPTAS for computing the CDTW distance between piecewise algebraic curves.  The algorithm runs in bit-complexity $O\left((m+n)^{\frac{19}{6}} (\frac{1}{\varepsilon})^{\frac{10}{3}} \log\frac{M+N}{\eps^2}\right)$, where $m$ and $n$ are the number of pieces of the input curves, see \Cref{FPTAS}.

We now give a brief summary of our structural results:
\begin{itemize}
      \item We characterize the shape of an optimal alignment for CDTW distance using Pontryagin's maximum principle; this is \Cref{optimalpath}.
    \item  We obtain a lower bound for CDTW distance of piecewise algebraic curves using their Fréchet distance; this is \Cref{lower-bound}.
    \item We develop a geometric pre-processing scheme based on average curvature (turning angles) \Cref{turningangles}. The goal of this pre-processing is to create principled and fast piecewise quadratic approximations of higher degree algebraic curves.
\end{itemize}
These structural results effectively create a piece-wise quadratic approximation with precise control on the error in CDTW distance. We expect these structural results to remain useful for future work to build on ours. 
\begin{itemize}
    \item Our main algorithmic engine is a dynamic program that computes CDTW between piecewise quadratic curves to any target multiplicative accuracy $\eps$.
    \item The idea of using dynamic programming and propagating cost through boundaries with functions rather than point-wise discretization is already present in \cite{buchin2022computing}. Our main contribution lies in the detailed approximation and bit-size  control machinery in \Cref{errorandcost} and \Cref{propagation}. It is worth mentioning that the tools used in these sections are Chebyshev interpolants (approximation theory) and resultants (real algebraic geometry) which are powerful but elementary tools in these respective fields. These basic tools  allows us to handle CDTW distance of non-linear curves.   
    \item Combining the structural results and the algorithmic engine for piecewise quadratics yields the FPTAS for any pair of piecewise algebraic curves  \Cref{FPTAS}.  
\end{itemize}

\section{Preliminaries}
\subsection{Setup and Definitions}
We have two piecewise algebraic curves $P$ and $Q$, divided into $m$ and $n$ pieces, respectively. Further, we assume that every piece is the image of a polynomial with two univariate polynomials of degree at most $d$ and coefficient bit-size at most $\tau$. This is just a simplifying assumption as polynomials approximate arbitrary algebraic curves.
\\\\
Let us define the CDTW distance between $P$ and $Q$. 
\begin{dfn}
Suppose the total arc-length of $P$ is $p$ and the total arc-length of $Q$ is $q$. 
Let $\Gamma(p,q) = \Gamma(p) \times \Gamma(q)$, where $\Gamma(p)$ is the set of all continuous, differentiable, and non-decreasing functions $\alpha: [0,1] \rightarrow [0, p]$ satisfying the boundary conditions $\alpha(0) = 0$ and $\alpha(1) = p$ and $\Gamma(q)$ is defined similarly.  We define the CDTW distance between curves $P$ and $Q$ as
\[d_{CDTW}(P,Q) = \inf_{(\alpha,\beta) \in \Gamma(p,q)}  \left( \int_0^1 ||P(\alpha(z)) - Q(\beta(z))||_2^2 \  \sqrt{\alpha'(z)^2 + \beta'(z)^2} \  dz. \right)^{\frac{1}{2}} \]
\end{dfn}

Our definition of CDTW distance differs from the one of \cite{buchin2022computing}. Our rationale is as follows:
\begin{itemize}
    \item We prefer to use $\ell_2$-arc-length for the curve $(\alpha,\beta)$ simply because $\ell_2$-arc-length is the basis for the classical differential geometry of curves and surfaces.
    \item In our view, the ``right'' intuitive way to think about CDTW distance is to consider  the cost of a continuous transportation of mass from curve $P$ to curve $Q$. With this view, Wasserstein distance is the natural fit. 
    \item After deciding on $\ell_2$-arc-length for the differential geometry, it is all but  natural to use squared-$\ell_2$-distance between the curves for the cost function (height function). The analog to this is the definition of Wasserstein-2 distance where  the integral has a square root outside. The square root after the integral that is customary in Wasserstein distance makes our CDTW definition to have the same ``unit'' with the aforementioned definitions.
     \item As discussed in \cite{klaren2020continuous}, the squared $\ell_2$ distance is also natural because that is the measure used in methods for center computation such as in k-means clustering as well as in the least-squares method for regression analysis.
\end{itemize}

We will map the CDTW computation from $\Gamma(p,q)$ to a rectangle parameter space. More precisely, we construct the parameter space $R = [0,p] \times [0,q]$ as an $(n-1)\times (m-1)$ grid of cells. See Figure 1.

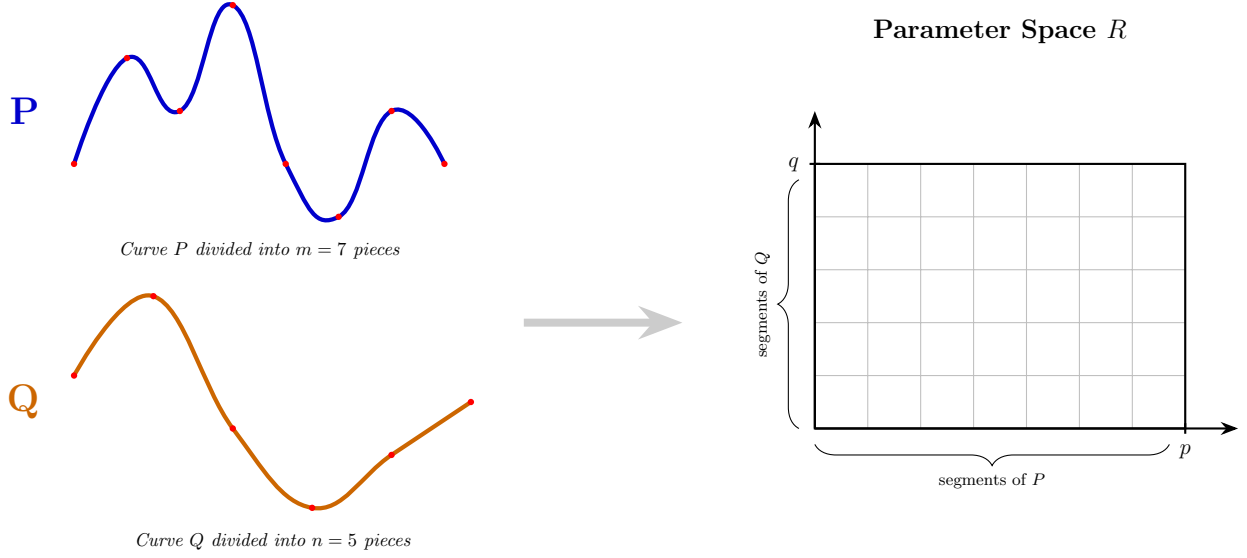
\begin{figure}[h]
\centering
\begin{tikzpicture}[
    scale=0.7, 
    transform shape,
    >=Stealth,
    critical point/.style={circle, fill=red, inner sep=1.2pt},
    curve label/.style={font=\huge\bfseries},
    axis label/.style={font=\large}
]

\begin{scope}[local bounding box=curveP]
    \draw[ultra thick, blue!80!black] plot [smooth, tension=0.8] coordinates {
        (0,0) (1,2) (2,1) (3,3) (4,0) (5,-1) (6,1) (7,0)
    };
    
    \node[critical point] at (0,0) {};
    \node[critical point] at (1,2) {};
    \node[critical point] at (2,1) {};
    \node[critical point] at (3,3) {};
    \node[critical point] at (4,0) {};
    \node[critical point] at (5,-1) {};
    \node[critical point] at (6,1) {};
    \node[critical point] at (7,0) {};
    
    \node[curve label, blue!80!black, left=15pt] at (0,1) {P};
    \node[below=10pt, font=\itshape] at (3.5,-1) {Curve $P$ divided into $m=7$ pieces};
\end{scope}

\begin{scope}[shift={(0,-5)}, local bounding box=curveQ]
    \draw[ultra thick, orange!80!black] plot [smooth, tension=0.7] coordinates {
        (0,1) (1.5,2.5) (3,0) (4.5,-1.5) (6,-0.5) (7.5,0.5)
    };
    
    \node[critical point] at (0,1) {};
    \node[critical point] at (1.5,2.5) {};
    \node[critical point] at (3,0) {};
    \node[critical point] at (4.5,-1.5) {};
    \node[critical point] at (6,-0.5) {};
    \node[critical point] at (7.5,0.5) {};
    
    \node[curve label, orange!80!black, left=15pt] at (0,0.5) {Q};
    \node[below=10pt, font=\itshape] at (3.75,-1.5) {Curve $Q$ divided into $n=5$ pieces};
\end{scope}

\draw[->, line width=3pt, gray!40] (8.5, -3) -- (11.5, -3) ;

\begin{scope}[shift={(14,-5)}]

    \draw[step=1cm, lightgray, thin] (0,0) grid (7,5);
    \draw[thick] (0,0) rectangle (7,5);
    
    \draw[->, thick] (0,0) -- (8,0);
    \draw[->, thick] (0,0) -- (0,6);
    
    \draw[thick] (7, 0.1) -- (7, -0.1) node[below=2pt, axis label] {$p$};
    \draw[thick] (0.1, 5) -- (-0.1, 5) node[left=2pt, axis label] {$q$};
    
    \draw[decorate, decoration={brace, amplitude=8pt, mirror}] (0,-0.3) -- (6.7,-0.3) 
        node[midway, below=12pt, font=\small] {segments of $P$};
        
    \draw[decorate, decoration={brace, amplitude=8pt}] (-0.3,0) -- (-0.3,4.7) 
        node[midway, left=12pt, rotate=90, anchor=south, font=\small] {segments of $Q$};
        
    \node[font=\Large\bfseries] at (3.5, 7.5) {Parameter Space $R$};
\end{scope}
\end{tikzpicture}
\caption{Mapping of curves to parameter space.}
\end{figure}

The curve $P$ has length $p$, and let the lengths of the pieces be $p_0,p_1,\ldots,p_{m-1}$ where $p=p_0+\cdots+p_{m-1}$ (similarly for $Q$ and $q$). Every piece of $P$ is given as an image of a polynomial map 
\[ (p_0+p_1+\cdots+p_{i-1}, p_0+p_1+\cdots+p_{i-1}+p_i )  \rightarrow \mathbb{R}^2 \] 
We name these maps as $P_0,P_1,\cdots,P_{m-1}$ and $Q_0,Q_1,\cdots,Q_{n-1}$ respectively where $P_i : \mathbb{R} \rightarrow \mathbb{R}^2$ (similarly $Q_i$) is given by a tuple of univariate polynomials. We assume those univariate polynomials have degree at most $d$ and bit-size $\tau$.  

In $R$ the cell $(i,j)$, with $i \in \{0, \ldots, n-1\}$ and $j \in \{0, \ldots, m-1\}$, corresponds to the comparison of $P_i$ to $Q_j$. For every point $(x,y) \in R$ we define a height function as follows.
\begin{dfn}
Let $(x,y) \in R$ be a point in cell $(i,j)$, then the height function $h$ is defined as follows:
\[ h(x,y) =  \norm{P_i(x)-Q_j(y)}_2^2 \]
\end{dfn}
We also define more terminology.
\begin{dfn}
    The left and bottom sides of a cell will be called its \textit{input boundaries} and the top and right sides of the cell its \textit{output boundaries}.
\end{dfn}
We can now re-formulate the CDTW distance as follows.
\begin{lemma}
\begin{equation}
    d_{CDTW}(P, Q) = \inf_{\gamma \in \Gamma(p,q)} \left( \int_0^{1} h(\gamma(z)) \ ||\gamma'(z)||_2 \  dz \right)^{\frac{1}{2}}.
\end{equation} 
\end{lemma}
This lemma essentially follows from definition but we write a proof for completeness.
\begin{proof}
    Recall Definition  of $d_{CDTW}(P,Q)$:
\[
d_{CDTW}(P,Q) = 
        \inf_{(\alpha,\beta) \in \Gamma(p,q)} \left( \int_0^1 ||P(\alpha(z)) - Q(\beta(z))||_2^2 \cdot \sqrt{\alpha'(z)^2 + \beta'(z)^2} \cdot dz \right)^{\frac{1}{2}} . \\
\]
Let $\gamma = (\alpha,\beta) \in \Gamma(p,q)$. Then $\gamma(0) = (\alpha(0), \beta(0)) = (0,0)$ and $\gamma(1) = (\alpha(1),\beta(1)) = (p,q)$. We can see that $\gamma(z)$ is a curve starting at $(0,0)$, ending at $(p,q)$, non-decreasing in both its $x$- and $y$-coordinates. Now, consider the integral of $h(\cdot)$ along the curve $\gamma$. This line integral $\int_\gamma h(z) \cdot dz$ is:

\[
    \int_\gamma h(z) \cdot dz = \int_0^1 h(\gamma(z)) \cdot ||\gamma'(z)||_2 \cdot dz.
\]
Hence,
\[
d_{CDTW}(P, Q) =  \inf_{\gamma \in \Gamma(p,q)} \left( \int_0^{1} h(\gamma(z)) \ ||\gamma'(z)||_2 \  dz \right)^{\frac{1}{2}}.  
\]
\end{proof}

Later in the paper we will need to have a running cost function which gives us the CDTW distance for the pair $(p,q)$. We define this as follows.
\begin{dfn}
For $(x,y) \in R$ let the arc-length of $P$ between $P(0)$ to $P(x)$ be $p_x$ and let the arc-length of $Q$ from $0$ to $y$ be $q_y$. Then, we define
\[ C(x,y) = \inf_{\Gamma(p_x,q_y)} \left( \int_{0}^1 h(\gamma(t)) \; \norm{\gamma'(t)} \; dt \right)^{\frac{1}{2}} \] 
\end{dfn}

We optimize by finding the best path $\gamma$ in the parameter region. However, there may be several paths which yield the CDTW distance. So, to make the phrase ``optimal path'' meaningful we come up with an arbitrary tie breaking rule.

\begin{dfn}  \label{tiebreakingrule}
    An optimal path is a path $\gamma = (\alpha, \beta)$ in the parameter space such that the cost function evaluated along $\gamma$ is minimum. If multiple curves minimize the cost function, the optimal path $\gamma^*$ is chosen to be the one that maximizes the area under the curve in the parameter space, defined mathematically as maximizing the integral $\int \beta(z) \alpha'(z) \, dz$ (the ``highest'' one).
\end{dfn}

\begin{lemma}
    Optimal path from $(0,0)$ to $(x,y) \in R$ is unique for all $(x,y)$.
\end{lemma}
\begin{proof}
    Suppose for the sake of contradiction that there exist two distinct optimal paths, $\gamma_1$ and $\gamma_2$, that both minimize the cost function and maximize the area integral as in \Cref{tiebreakingrule}. Since they are distinct continuous and non-decreasing paths, there must exist a region where one path is strictly above the other. Let $\gamma_{\max}$ be the path defined by the pointwise maximum (upper envelope) of $\gamma_1$ and $\gamma_2$. The path $\gamma_{\max}$ will have a cost less than or equal to the minimum cost, because it consists entirely of segments of the two optimal paths meeting at their intersection points, its accumulated cost does not exceed the minimal cost. Thus it also an optimal path. Furthermore, the area under $\gamma_{\max}$ is strictly greater than the area under $\gamma_1$ and $\gamma_2$, contradicting the assumption that $\gamma_1$ and $\gamma_2$ maximized the area integral. Thus exactly one optimal path exists, and is thus unique, in the parameter space of two curves $P$ and $Q$.
\end{proof}

Finally, we define the concepts of a ``lower envelope" and a ``cumulative minimum" to be used later.

\begin{dfn}[Lower Envelope]
    The lower envelope of a set of functions is the function that results by taking, at each point on the $x$-axis, the point on at least one of the functions that has the least $y$-coordinate. See Figure 2 for an example. Note that the functions for which we apply the lower envelope below are all always nonnegative, so we needn't deal with positive/negative $y$-coordinates.
\end{dfn}
\begin{figure}[h]
\centering
\begin{tikzpicture}
\begin{axis}[
    axis lines = middle,
    xlabel = $x$,
    ylabel = $y$,
    xmin=0, xmax=5,
    ymin=0, ymax=6,
    xtick=\empty, ytick=\empty,
    legend pos=outer north east,
    width=8cm, 
    height=6cm 
]
\addplot [
    domain=0:5, 
    samples=100, 
    color=blue,
    thick,
    opacity=0.3
] {0.5*(x-2)^2 + 1.5};
\addlegendentry{$f(x)$}

\addplot [
    domain=0:5, 
    samples=100, 
    color=green!60!black,
    thick,
    opacity=0.3
] {-0.8*x + 5};
\addlegendentry{$g(x)$}

\addplot [
    domain=0:5, 
    samples=100, 
    color=orange,
    thick,
    opacity=0.3
] {sin(deg(x*2)) + 3};
\addlegendentry{$h(x)$}

\addplot [
    domain=0:5, 
    samples=300,
    color=red,
    ultra thick,
    dashed
] {min(min(0.5*(x-2)^2 + 1.5, -0.8*x + 5), sin(deg(x*2)) + 3)};
\addlegendentry{Lower Envelope}

\end{axis}
\end{tikzpicture}
\caption{Example of a lower envelope.}
\end{figure}
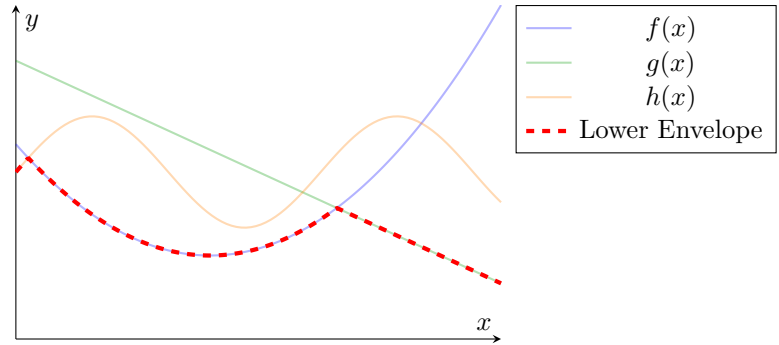

\begin{dfn}[Cumulative Minimum]
    The cumulative minimum of a function $f(x)$ over an interval $[a,b]$ is the non-increasing function $g(b) = \inf\{f(x): a \le x \le b\}$, representing the lower value achieved by $f$ in the interval.
\end{dfn}

\subsection{Structure of the Valley}

The points where the height function vanishes play a special role, aptly named the ``valley''. 
\begin{dfn}
The collection of points $(x,y) \in R$ such that $h(x,y)=0$ is called the valley. 
\end{dfn}
Now suppose $(x,y)$ is in cell $(i,j)$ and $P_i(x)=(x,u(x))$ and $Q_j(y)=(y,v(y))$ where $u,v$ are univariate polynomials. We then have the following description: $(x,y)$ is in the valley if and only if 
\[ h(x,y) = (u(x)-v(y))^2 + (x-y)^2 = 0  \Leftrightarrow x=y \; \text{and} \; u(x)-v(x) = 0 \]
This means the valley can come in two shapes:
\begin{itemize}
    \item $x-y | u(x)-v(y)$ or equivalently $u(x)-v(x)$ is the zero polynomial. In this case, the valley is the $x=y$ line in the cell $(i,j)$. See Figure 3A.
    \item $x- y$ is not a common factor. In this case, the valley is depicted by solutions of $u(x)-v(x)=0$. The valley is the collection of at most $d$ points $(x,x)$ where $x$ satisfies $u(x)-v(x)=0$. See Figure 3B.
\end{itemize}
We can solve the equation in the second case in $O(d^2 \tau)$ time and it has at most $d$ many solutions (see section 1.2 of \cite{Ergur_2022} for a synopsis of current algorithms).

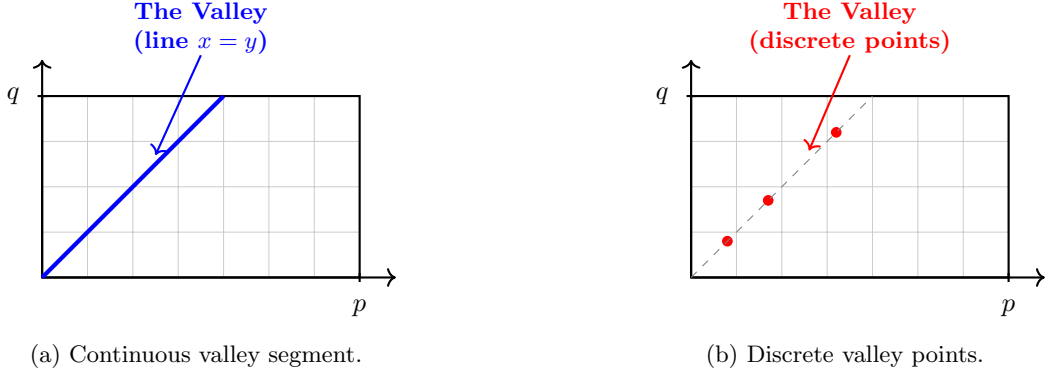
\begin{figure}[h]
  \centering
  \begin{subfigure}[b]{0.48\textwidth}
    \centering
    \begin{tikzpicture}[scale=0.6]
      \draw[step=1cm, gray!40, very thin] (0,0) grid (7,4);
      \draw[thick] (0,0) rectangle (7,4);
      \draw[->, thick] (0,0) -- (7.8,0);
      \draw[->, thick] (0,0) -- (0,4.8);
      \draw[thick] (7, 0.1) -- (7, -0.1) node[below=3pt] {$p$};
      \draw[thick] (0.1, 4) -- (-0.1, 4) node[left=3pt] {$q$};
      \draw[blue, ultra thick] (0,0) -- (4,4);
      \node[blue, font=\small\bfseries, align=center] at (3.5, 5.5) {The Valley\\(line $x=y$)};
      \draw[->, blue, thick] (3.5, 4.9) -- (2.5, 2.7);
    \end{tikzpicture}
    \caption{Continuous valley segment.}
  \end{subfigure}
  \hfill
  \begin{subfigure}[b]{0.48\textwidth}
    \centering
    \begin{tikzpicture}[scale=0.6]
      \draw[step=1cm, gray!40, very thin] (0,0) grid (7,4);
      \draw[thick] (0,0) rectangle (7,4);
      \draw[->, thick] (0,0) -- (7.8,0);
      \draw[->, thick] (0,0) -- (0,4.8);
      \draw[thick] (7, 0.1) -- (7, -0.1) node[below=3pt] {$p$};
      \draw[thick] (0.1, 4) -- (-0.1, 4) node[left=3pt] {$q$};
      \filldraw[red] (0.8, 0.8) circle (3pt);
      \filldraw[red] (1.7, 1.7) circle (3pt);
      \filldraw[red] (3.2, 3.2) circle (3pt);
      \draw[gray, dashed, thin] (0,0) -- (4,4);
      \node[red, font=\small\bfseries, align=center] at (3.5, 5.5) {The Valley\\(discrete points)};
      \draw[->, red, thick] (3.5, 4.9) -- (2.6, 2.8);
    \end{tikzpicture}
    \caption{Discrete valley points.}
  \end{subfigure}
  \caption{Structure of valley per-cell.}
\end{figure}

\subsection{Optimal Path Disjointness}
Here we note an observation to be used later.
\begin{lemma} \label{path_spatiality}
    Two optimal paths with different starting points or different end points do not intersect.
\end{lemma}
\begin{proof}
Suppose $\gamma_1,\gamma_2$ are two optimal paths inside a cell. $\gamma_1$ starts from $X_1$ and ends at $Z_1$ and $\gamma_2$ starts at $X_2$ and ends at $Z_2$ where $X_1 \neq X_2$. We claim $\gamma_1$ and $\gamma_2$ do not intersect. To prove the claim suppose $\gamma_1$ and $\gamma_2$ first intersect at point $Z_0$. Then, either the path $\gamma_1$ followed to arrive to $Z_0$ is optimal or the path $\gamma_2$ followed to arrive to $Z_0$ (and not both). This gives a contradiction and establishes the claim.
\end{proof}

\subsection{Subdivision of Boundaries with Elimination Theory} \label{realag}
Our optimal cost computations will essentially be a polynomial optimization task. This follows from the definition of our height function.  In this section we will derive some general techniques for subdividing input and output boundaries under polynomial cost functions. The main purpose of this section is to introduce basic tools from real algebraic geometry that may not be familiar to some of our readers \cite{basu2006algorithms}.

Suppose for $s \in [0,1]$ and $t \in [0,1]$ the optimal cost of traveling from $s$ to $t$ is given by a polynomial $f(s,t)$ of degree $d$. For a given $t_0$ how do we find $\operatorname*{argmin}_{s}f(s,t_0)$?  We start by identifying the critical points $t$ where there is an abrupt jump in the location of minimizer $s$. These jumps happen when two branches of minimizers collide.
\begin{enumerate}
\item The algebraic expression for this is the following:
\[ \frac{\partial}{\partial s} f(s_1 ,t) = 0 \; , \; \frac{\partial}{\partial s} f(s_2 ,t) = 0 \; , \; f(s_1,t) - f(s_2,t) = 0.\]
The last equation is divisible by $s_1-s_2$ and since $s_1 \neq s_2$ we factor that part out. 
\item We now have three equations of degree $d-1$. We will eliminate $s_1$ and $s_2$ and obtain an equation in only $t$. We first compute the resultant of first and last equation to eliminate $s_1$. The resultant, denoted $R(s_2,t)$, has degree at most $(d-1)^2$. Now we compute the resultant of $R(s_2,t)$ with the second equation; this gives a polynomial $Q(t)$ in $t$ with degree at most $(d-1)^3$.
\item We solve $Q(t)$ to compute the critical points $t$.
\end{enumerate}

In the second phase of our computation we subdivide $[0,1]$ using the $(d-1)^3$ critical points that are computed. This stratum has less than $d^3$ many critical points and  less than $d^3$ many open intervals in it. Denote these intervals as $I_1,I_2,\cdots,I_k$. For every $t \in I_j$, the minimizer $s$ is unique; let's denote $\phi(t)$ to be the map that sends $t$ to minimizer $s$. Note that for any $t \in I_k$, the map $\phi(t)$ satisfies
\[  \frac{\partial }{\partial s} f(\phi(t),t) = 0.\]
Moreover, the map $\phi(t)$ is differentiable since the interval contains no critical points. On the critical points, the map $\phi$ remains continuous. This implies that the image of $I_k$ under $\phi$ is an interval.

\section{The Shape of an Optimal Path} \label{optimalpath}
In this section we'll characterize the shape of optimal paths using Pontryagin's Maximum Principle (\cite{pontryagin1962}, \cite{liberzon2011calculus}). We desire to find a curve $\gamma$ such that
\begin{itemize}
\item $\gamma$ is non-decreasing in both $x$ and $y$ coordinates
\item $\gamma(0)=(0,0)$ and $\gamma(1) = (p,q)$ where $p$ and $q$ are the lengths of the polynomial curves.
\item $\gamma$ minimizes the cost function $\int_{0}^1 h(\gamma(t)) \norm{\gamma'(t)}_2 \; dt $
\end{itemize}
Here the integral being between $0$ and $1$ will not effect any of our arguments and can be easily adjusted to other start-end points. The first condition gives us two constraints; both coordinates of $\gamma'(t)$ needs to be non-negative all the time. So, we define the following augmented Hamiltonian:
\[  H(t,z, u, p) = h(z)  \norm{u}_2  + \langle p , u \rangle - \lambda_1 u_1 - \lambda_2 u_2  \]
where $z(t)=\gamma(t)$ and $u(t) =\gamma'(t) $. 

Here $\lambda_1(t) > 0$ whenever $u_1(t) < 0$ and $\lambda_1(t) = 0$ whenever $u_1(t) \geq 0$. $\lambda_2(t)$ and $u_2(t)$ observes the same pattern. We say $\lambda_i$ is active when we have $\lambda_i \neq 0$. Now we write down Hamilton's equations:
\[  u = H_p  \;, \;   p' = - H_z = - \norm{u}_2 \nabla_z h(z),    \]
where $\nabla_z h(z)$ is defined as long as the height function is differentiable. Due to Pontryagin's Maximum Principle we must have 
\[  0 = H_u =h(z) \frac{u}{\norm{u}_2} + p - (\lambda_1 ,\lambda_2) \]
Thus
\[  p =  (\lambda_1 , \lambda_2 ) -  h(z) \frac{u}{\norm{u}_2}  \Rightarrow p' = (\lambda_1', \lambda_2') - \frac{d}{dt} \left( h(z) \frac{u}{\norm{u}_2}  \right) \]

So, we have
\begin{equation} \label{optimalityequation}
  - \norm{u}_2 \nabla_z h(z)  =  (\lambda_1', \lambda_2') - \frac{d}{dt} \left( h(z) \frac{u}{\norm{u}_2} \right)  =  (\lambda_1', \lambda_2')  -  h(z) \frac{d}{d t} \left( \frac{u}{\norm{u}_2} \right) - h'(z) \frac{u}{\norm{u}_2}
\end{equation}

We note that 
\[ h'(z) = \nabla_z h(z) u^T  \Rightarrow  h'(z) \frac{u}{\norm{u}_2} =   \nabla_z h(z) \norm{u}_2 \]
Using this inside \Cref{optimalityequation} we have
\begin{equation} \label{generalmain}
   h(z) \frac{d}{d t} \left( \frac{u}{\norm{u}_2} \right) =   (\lambda_1', \lambda_2') 
 \end{equation}
 
Now we will do case-by-case analysis: For the case of both $\lambda_1, \lambda_2$ being inactive, that is for $u_1 \geq 0$ and $u_2 \geq 0$ this gives us the following:
\begin{equation}  \label{positive}
  h(z) \frac{d}{d t} \left( \frac{u}{\norm{u}_2} \right) = (0,0)
\end{equation} 
Which implies,
\begin{equation} \label{hmm}
\forall t  \;  \text{either} \; h(\gamma(t))=0 \; \text{or} \;  \frac{d}{d t} \left( \frac{u}{\norm{u}_2} \right) = (0,0)
\end{equation}
This means either an optimal curve $\gamma$ is included in the zero set
\[  Z(h) := \{ z \in \mathbb{R}^2 : h(z) = 0   \} \]
or 
\[ \frac{d}{d t} \left( \frac{u}{\norm{u}_2} \right) = (0,0) \]
On $Z(h)$ the cost is zero by definition, so the curve included in $Z(h)$ does not impact the cost. This is the collection of points that we called ``valley''.  If the curve is not included in $Z(h)=0$, then either $\gamma_1'(t) =0$ (vertical line) or $\gamma_2'(t)=0$ (horizontal line) or from \ref{hmm} we have we have
 \[  (0,0) = \frac{d}{d t} \left( \frac{u}{\norm{u}_2} \right)  = \frac{d}{d t} \left( \frac{\gamma_1'(t)}{\norm{\gamma'(t)}_2} , \frac{\gamma_2'(t)}{\norm{\gamma'(t)}_2} \right)  \] 
  
\begin{equation*}
(0,0)  = \left( \frac{\gamma_1''(t)}{\norm{\gamma'}_2} - \frac{\gamma_1'(t) (\gamma_1'(t) \gamma_1''(t) + \gamma_2'(t) \gamma_2''(t) )}{\norm{\gamma'}_2^3 } ,  \frac{\gamma_1''(t)}{\norm{\gamma'}_2} - \frac{\gamma_2'(t) (\gamma_1'(t) \gamma_1''(t) + \gamma_2'(t) \gamma_2''(t) )}{\norm{\gamma'}_2^3 }   \right)
\end{equation*}  
Note that   $\gamma_1'(t) \gamma_1''(t) + \gamma_2'(t) \gamma_2''(t)  = \frac{d}{dt} \left( \frac{1}{2} \norm{\gamma'(t)}_2^2 \right)$. So, we have
\begin{equation*} 
  \frac{\gamma_1''(t)}{\norm{\gamma'(t)}}  =   \frac{\gamma_1'(t)}{\norm{\gamma'(t)}^{3}} \frac{d}{dt} \left( \frac{1}{2} \norm{\gamma'(t)}^2 \right) \; \text{and} \;  \frac{\gamma_2''(t)}{\norm{\gamma'(t)}}  =   \frac{\gamma_2'(t)}{\norm{\gamma'(t)}^{3}} \frac{d}{dt}  \left( \frac{1}{2} \norm{\gamma'(t)}^2 \right)
\end{equation*}
Written more concisely, we have
\begin{equation}\label{hmmm}
\frac{\gamma_1''(t)}{\gamma_1'(t)} = \frac{\gamma_2''(t)}{\gamma_2'(t)} = \frac{1}{\norm{\gamma'(t)}^2} \frac{d}{dt} \left( \frac{1}{2} \norm{\gamma'(t)}^2 \right)
\end{equation}
We denote the curvature by $\kappa(t)$, and \ref{hmmm} yield the following
\begin{equation}
\kappa(t) = \frac{|\gamma_1'(t) \gamma_2''(t) -  \gamma_2'(t) \gamma_1''(t)|}{\norm{\gamma'(t)}^{\frac{3}{2}}} = 0
\end{equation}
In this case the optimal curve has to be a line.   

For the cases where only one constraint is active, that is, in the cases where $\gamma_1'(t)=0 , \gamma_2'(t) >0$ or $\gamma_1'(t) >0, \gamma_2'(t)=0$ the curve is a vertical or horizontal line. We analyze the case $\gamma_1'(t) >0, \gamma_2'(t)=0$ and the other case is analogous: from \ref{generalmain} we have
\begin{equation}
h(z) \frac{d}{dt} \left( \frac{\gamma_1'(t)}{\gamma_1'(t)} , 0   \right) = h(z) (0,0)  = (0, \lambda_2')
\end{equation}
which is satisfied vacuously.
\\ \\
Thus an optimal curve $\gamma$ consists of pieces that are either a line segment or it is included in the valley $Z(h)$, and in both cases we have $\gamma_1(t) \geq 0$ and $\gamma_2(t) \geq 0$ for every $t$.

\begin{rem}
The proof above works directly for the height function  
\[ h(x,y)= \sqrt{(P(x)-Q(y))^2 + (x-y)^2} \] with the only difference being the loss of differentiability when path touches the valley. Since the proof allows change of direction when the path touches the valley, this does not change anything in the conclusion.  For the cases $h(x,y)= \left( |P(x)-Q(y)|^p + |x-y|^p \right)^{\frac{1}{p}}$ with odd $p$, differentiability can be lost either when $P(x)=Q(y)$ or $x=y$. The conclusion of the proof for such $\ell_p$ norms would be that the optimal path is a line segment between joins and those joins would located at the intersection of the path either with $P(x)=Q(y)$ or $x=y$.
\end{rem}

\begin{rem} \label{threeoptions}
 In summary, an optimal path inside a cell can have three shapes: 
 \begin{enumerate}
     \item A line segment that never changes its direction until touching the output boundary,
     \item a line segment that travels towards the valley, arrives to the valley only  to realize the valley is the entire diagonal $x=y$, line segment travels inside the valley and leaves it an optimal point changing its slope,
     \item a line segment that travels towards the valley where the valley is a collection of discrete points, enters the diagonal $x=y$ direction from one of those points and leaves with a different direction at another valley point.
 \end{enumerate}
\end{rem}
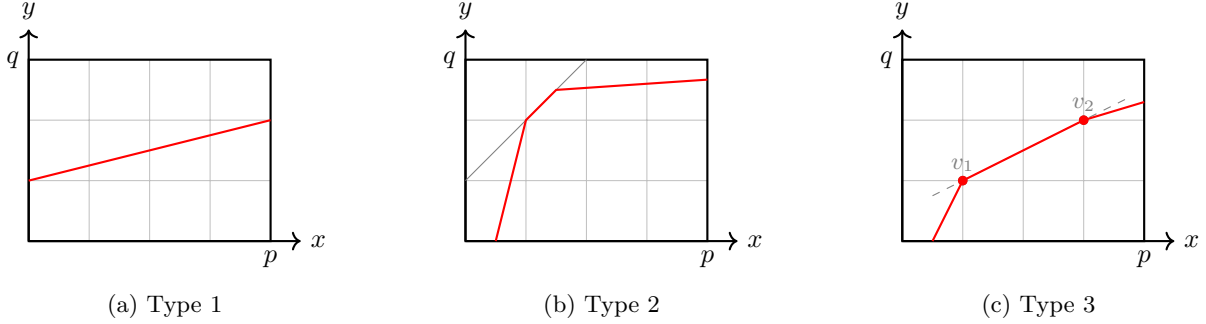
\begin{figure}[h] 
  \centering
  
  \begin{subfigure}[b]{0.3\textwidth}
    \centering
    \begin{tikzpicture}[scale=0.8]

      \draw[step=1cm, gray!50, very thin] (0,0) grid (4,3);
      \draw[thick] (0,0) rectangle (4,3);
      
      \draw[->, thick] (0,0) -- (4.5,0) node[right] {$x$};
      \draw[->, thick] (0,0) -- (0,3.5) node[above] {$y$};
      \node[anchor=north] at (4,0) {$p$};
      \node[anchor=east] at (0,3) {$q$};

      \draw[red, thick] (0, 1) -- (4, 2);
    \end{tikzpicture}
    \caption{Type  1}
  \end{subfigure}
  \hfill
  \begin{subfigure}[b]{0.3\textwidth}
    \centering
    \begin{tikzpicture}[scale=0.8]
      \draw[step=1cm, gray!50, very thin] (0,0) grid (4,3);
      \draw[thick] (0,0) rectangle (4,3);
      
      \draw[->, thick] (0,0) -- (4.5,0) node[right] {$x$};
      \draw[->, thick] (0,0) -- (0,3.5) node[above] {$y$};
      \node[anchor=north] at (4,0) {$p$};
      \node[anchor=east] at (0,3) {$q$};

      \draw[gray, thin] (0, 1) -- (2, 3);
      \draw[red, thick] (0.5, 0) -- (1, 2) -- (1.5, 2.5) -- (4, 2.67);
    \end{tikzpicture}
    \caption{Type 2}
  \end{subfigure}
  \hfill
  \begin{subfigure}[b]{0.3\textwidth}
    \centering
    \begin{tikzpicture}[scale=0.8]

      \draw[step=1cm, gray!50, very thin] (0,0) grid (4,3);
      \draw[thick] (0,0) rectangle (4,3);
      
      \draw[->, thick] (0,0) -- (4.5,0) node[right] {$x$};
      \draw[->, thick] (0,0) -- (0,3.5) node[above] {$y$};
      \node[anchor=north] at (4,0) {$p$};
      \node[anchor=east] at (0,3) {$q$};

      \filldraw[gray] (1, 1) circle (2pt) node[anchor=south, font=\small] {$v_1$};
      \filldraw[gray] (3, 2) circle (2pt) node[anchor=south, font=\small] {$v_2$};
      \draw[gray, dashed, thin] (0.5, 0.75) -- (3.75, 2.375);

      \draw[red, thick] (0.5, 0) -- (1, 1) -- (3, 2) -- (4, 2.3);
      \filldraw[red] (1, 1) circle (2pt);
      \filldraw[red] (3, 2) circle (2pt);
    \end{tikzpicture}
    \caption{Type 3}
  \end{subfigure}
  \caption{Anatomy of possible optimal paths.}
  \label{fig6}
\end{figure}

\section{Comparing CDTW and Fréchet Distance for Algebraic Curves} \label{lower-bound}
Previous sections established the shape of optimal path: it is a piecewise-linear function with joints either on the boundaries or on the valley. A consequence of this is that $h(\gamma(t))$ is a polynomial in $t$ on every linear piece of $\gamma(t)$. Now, we will utilize this to give a lower bound for CDTW distance in terms of Fréchet distance. Recall our assumption that we are given piecewise algebraic curves $P$ and $Q$, where $P, Q$ have $m,n$ pieces respectively, and the polynomials defining these pieces have degree at most $d$ and absolute value any coefficient at most $C$. Rote's algorithm \cite{rote2007computing} (and \cite{Conradi_Driemel_Kolbe_2025} for the decision-version of the problem) can be used to compute Fréchet distance of $P$ and $Q$ in time $O(mn \log(mn))$ time (ignoring dependence on $d$ and $C$). We will use these fast algorithms to get an upper and lower bound for $CDTW(P,Q)$.

\begin{prop}
Let $P$ and $Q$ be two curves as described above with total arc-lengths $p$ and $q$, then we have
\[   \frac{\mathrm{Frechet}(P,Q)^2}{\sqrt{2} d \sqrt{C}} \leq \mathrm{CDTW}(P,Q) \leq \mathrm{Frechet}(P,Q)\sqrt{p+q} \]
Furthermore, if the two curves $P$ and $Q$ do not intersect then we have
\[   \frac{\mathrm{Frechet}(P,Q)}{2\sqrt{3}d} \leq \mathrm{CDTW}(P,Q) \leq \mathrm{Frechet}(P,Q)\sqrt{p+q} \]
\end{prop}

\begin{proof}
The upper bound follows directly from the definition: taking the path $\gamma$ that realizes the Fréchet distance, the height function is bounded by $\text{Frechet}(P,Q)^2$, and integrating this constant over the path length (which is bounded by $p+q$) yields $\text{CDTW}(P,Q) \le \text{Frechet}(P,Q)\sqrt{p+q}$. 

Thus, we focus on the lower bounds. We focus on the case $\text{Frechet}(P, Q) \le 1$ as the other case is easier to prove. Let $\gamma$ be the optimal path that yields $\text{CDTW}(P, Q)$, and let $t_1 = \text{argmax}_t h(\gamma(t))$ be the time where the height is maximized on the path $\gamma$. We have $h(\gamma(t_1)) \ge \text{Frechet}(P, Q)^2$ and $h'(\gamma(t_1)) = 0$. 

Our parameter regime is divided into rectangles where the definition of $h$ changes in every rectangle, and the curve $\gamma$ travels from the bottom-left to the top-right. Suppose the curve $\gamma$ is in the same rectangle as $\gamma(t_1)$ for $t \in [t_0, t_2]$. Using Markov's inequality for polynomials, we have
$$ \max_{t \in [t_0, t_2]} h''(\gamma(t)) \le \frac{4(2d)^4}{3(t_2-t_0)^2} \max_{t \in [t_0, t_2]} h(\gamma(t)) = \frac{4(2d)^4}{3(t_2-t_0)^2} h(\gamma(t_1)).$$ 
We also have a universal upper bound $\max h''(\gamma(t)) \le d^2C$: we are interested in the case where $h(\gamma(t_1)) \le 1$. Using Taylor's theorem, for every $t \in [t_0, t_2]$ wherever $|t - t_1| \le s$ we have
$$ h(\gamma(t)) \ge h(\gamma(t_1)) - \frac{s^2}{2} \max h''(\gamma(t)) \ge h(\gamma(t_1)) - \frac{s^2 4(2d)^4 h(\gamma(t_1))}{6(t_2-t_0)^2}.$$ 
Thus, for every $t$ such that $|t - t_1| \le s$, we have 
$$ h(\gamma(t)) \ge \text{Frechet}(P, Q)^2 \left( 1 - \frac{32d^4}{3(t_2-t_0)^2} s^2 \right).$$ 
Specifically, for $s \le \frac{t_2-t_0}{4d^2}$ and $|t - t_1| \le s$ we have $h(\gamma(t)) \ge \frac{1}{3}\text{Frechet}(P, Q)^2$.
\\
The shape of $\gamma$ for $t \in [t_0, t_2]$ can take one of the following forms:
\begin{enumerate}
    \item $\gamma$ never touches the valley and thus $\|\gamma'(t)\|$ is constant for all $t \in [t_0, t_2]$,
    \item $\gamma$'s last touch to the valley, $t'_0$, is before $t_1$ and $\|\gamma'(t)\|$ is a constant afterwards,
    \item $\gamma$'s first touch to the valley, $t'_1$, is after $t_1$ and $\|\gamma'(t)\|$ is a constant before $t'_1$.
\end{enumerate}

In the first case, $\|\gamma'(t)\| \ge \frac{1}{t_2-t_0} \ge 1$ since the arc-length of $\gamma$ inside the rectangle is more than one. Integrating over the interval where the Taylor bound holds:
$$ \int_{t_0}^{t_2} h(\gamma(t))\|\gamma'(t)\| dt \ge \text{Frechet}(P, Q)^2 \int_{\max\{t_0, t_1 - \frac{t_2-t_0}{4d^2}\}}^{\min\{t_2, t_1 + \frac{t_2-t_0}{4d^2}\}} \left( 1 - \frac{32d^4}{3(t_2-t_0)^2} (t - t_1)^2 \right) \frac{1}{t_2 - t_0} dt $$

$$ \Longleftrightarrow \int_{t_0}^{t_2} h(\gamma(t))\|\gamma'(t)\| dt \ge \text{Frechet}(P, Q)^2 \frac{\min\{t_2, t_1 + \frac{t_2-t_0}{4d^2}\} - \max\{t_0, t_1 - \frac{t_2-t_0}{4d^2}\}}{3(t_2 - t_0)},$$
which evaluates to at least $\frac{\text{Frechet}(P, Q)^2}{12d^2}$. Taking the square root yields:
$$ \text{CDTW}(P, Q) \ge \left( \int_{t_0}^{t_2} h(\gamma(t))\|\gamma'(t)\| dt \right)^{1/2} \ge \frac{\text{Frechet}(P, Q)}{\sqrt{12} d} = \frac{\text{Frechet}(P, Q)}{2\sqrt{3}d}.$$
We note that if the two curves $P$ and $Q$ do not intersect, there is no valley and we are always in this first case. 

In the second case, $t_2 - t_0 \ge t_1 - t'_0 \ge \frac{t_2-t_0}{4d^2}$, since $h(\gamma(t'_0)) = 0$. Note that the Lipschitz constant of $h(x, y)$ on the rectangle is at most $d^2 C$, yielding a minimal spatial distance:
$$ \|\gamma(t'_0) - \gamma(t_1)\| \ge \frac{\text{Frechet}(P, Q)^2}{d^2C} $$ 
Because $h$ drops from $\text{Frechet}(P,Q)^2$ to $0$ with a maximum slope of $d^2C$, the integral of $h$ over this segment is bounded below by the area of a triangle with height $\text{Frechet}(P,Q)^2$ and base $\frac{\text{Frechet}(P,Q)^2}{d^2C}$:
$$ \int_{t'_0}^{t_1} h(\gamma(t))\|\gamma'(t)\| dt \ge \frac{1}{2} \left(\text{Frechet}(P,Q)^2\right) \left( \frac{\text{Frechet}(P,Q)^2}{d^2C} \right) = \frac{\text{Frechet}(P, Q)^4}{2d^2C}.$$
Taking the square root yields the general lower bound:
$$ \text{CDTW}(P,Q) \ge \frac{\text{Frechet}(P,Q)^2}{\sqrt{2} d \sqrt{C}} $$
The proof for the third case is identical.
\end{proof}

\section{Piecewise Quadratic Approximation via Turning Angles} \label{turningangles}

Rote, who computes Fréchet distance between smooth curves, used curvature to create local approximation with circles and off-set curves \cite{rote2007computing}. CDTW distance is in some sense an integral Fréchet distance. Thus, Rote's work suggests the integral of curvature is a good local approximation measure for our purposes. Luckily, integral of curvature is nothing but the total turning angle. We construct piecewise quadratic approximation using Bezier curves and total turning angles as follows.

\begin{enumerate}
    \item For a given algebraic curve $(x, p(x))$ with $x \in I$: Compute critical points of $p(x)$, $p'(x)$, $p''(x)$ on $I$, and subdivide $I$ into $I_1,I_2,\ldots,I_k$ using these critical points.
    \item For each $I_i=[a_i,b_i]$, construct the quadratic Bézier curve $\tilde{p}_k$ with endpoints $C_0 = (a_i,p(a_i))$ and $C_2 = (b_i,p(b_i))$.  $C_1$ is the unique intersection of the tangent lines at $C_0$ and $C_2$.
\end{enumerate}

Note that in every sub-interval the polynomial $p(x)$ is strictly convex or concave, and the curvature does not change sign. This allows us to derive the following distance bound.

\begin{lemma} \label{lem:spatial_bound}
Let $\tilde{p}_k$ and $I_k$ be as above, and let  $\ell_k = \norm{C_2 - C_0}_2$  be the chord-length and  $\Theta_k = \int_{I_k} |\kappa(s)| \, ds < \frac{\pi}{2}$ total turning angle. Then, we have
\[ \sup_{x \in I_k} \norm{\tilde{p}_k(x) - p(x)}_2^2 \le \frac{1}{4} \ell_k^2 \tan\left(\frac{\Theta_k}{2}\right)^2. \]
\end{lemma}

We first give a proof by picture, followed by a more formal proof.

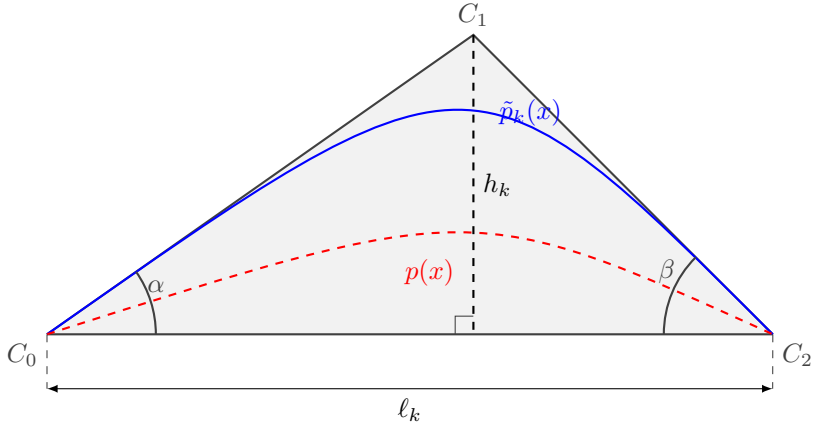
\begin{tikzpicture}[scale=1.2, >=latex, font=\sffamily]

    \coordinate (C0) at (0,0);
    \coordinate (C2) at (8,0);
    \coordinate (C1) at (4.7, 3.3); 
    \coordinate (H)  at (4.7, 0); 

    \fill[gray!10] (C0) -- (C1) -- (C2) -- cycle;

    \draw[thick, darkgray] (C0) -- (C1) -- (C2);
    \draw[thick, darkgray] (C0) -- (C2);

    \draw[dashed, thick, black] (C1) -- (H) node[midway, right] {$h_k$};
    
    \draw[darkgray] (4.7, 0.2) -- (4.5, 0.2) -- (4.5, 0);

    \draw[thick, blue] (C0) .. controls (C1) .. (C2) 
        node[pos=0.55, above=8pt, right=2pt, blue] {$\tilde{p}_k(x)$};
        
    \draw[thick, red, dashed] (C0) .. controls (4.7, 1.5) .. (C2) 
        node[pos=0.45, below=6pt, red] {$p(x)$};

    \draw[thick, darkgray] (1.2,0) arc (0:35:1.2) 
        node[midway, right=2pt, above] {$\alpha$};
        
    \draw[thick, darkgray] (6.8,0) arc (180:135:1.2) 
        node[midway, left=2pt, above] {$\beta$};

    \node[below left, darkgray] at (C0) {$C_0$};
    \node[above, darkgray] at (C1) {$C_1$};
    \node[below right, darkgray] at (C2) {$C_2$};

    \draw[<->] (0, -0.6) -- (8, -0.6) node[midway, below] {$\ell_k$};
    \draw[dashed, darkgray] (C0) -- (0, -0.6);
    \draw[dashed, darkgray] (C2) -- (8, -0.6);

    \node[anchor=north west, text width=6cm, align=left] at (0, 4.2) {
    };

\end{tikzpicture}

\begin{proof}
Both the polynomial and the quadratic Bezier approximation  are strictly convex (concave) on $I_k$ and lie entirely within the tangent triangle $\triangle C_0 C_1 C_2$. Let $\alpha = \angle C_1 C_0 C_2$ and $\beta = \angle C_1 C_2 C_0$, note that $\alpha + \beta \le \Theta_k$. The distance between $\tilde{p}_k$ and $(x,p(x))$ is bounded by the altitude $h_k$ of vertex $C_1$ over $\overline{C_0 C_2}$. Note that
\[ \norm{C_1 - C_0}_2 = \ell_k \frac{\sin \beta}{\sin(\alpha + \beta)} \] 
Thus, the altitude is
\[
    h_k = \norm{C_1 - C_0}_2 \sin \alpha = \ell_k \frac{\sin \alpha \sin \beta}{\sin(\alpha + \beta)}.
\]
Applying the identity $\sin \alpha \sin \beta = \frac{1}{2}[\cos(\alpha - \beta) - \cos(\alpha + \beta)]$ and noting that $\cos(\alpha - \beta) \le 1$, we obtain
\[
    h_k \le \ell_k \frac{1 - \cos(\alpha + \beta)}{2 \sin(\alpha + \beta)} = \frac{1}{2} \ell_k \tan\left(\frac{\alpha + \beta}{2}\right).
\]
Finally, note that $\tan(x/2)$ is strictly increasing on $\left(0, \frac{\pi}{2}\right)$ and $\alpha + \beta \le \Theta_k$, giving the claimed bound.
\end{proof}

We will subdivide further (if needed) to make sure the error bound in Lemma \ref{lem:spatial_bound}  is always less than $\varepsilon^2$ for a target $\varepsilon$ on every quadratic piece. Let us first write an easy bound on the turning angle cuts.
\begin{prop}
Let $(x,p(x))$ be a curve where $x \in [0,1]$ and $p$ is a degree $d$ curve. $p$ can be subdivided into at most $\frac{(d-1)\pi}{\arctan(\varepsilon)}$ pieces where every piece has total turning angle at most $\min \{ \frac{\pi}{2} , \arctan(\varepsilon) \}$.
\end{prop}
We skip the proof of this claim as it is standard.  We focus on algorithmic aspects: Our first step is to create an exact computation predicate that ensures the expression $\frac{1}{4} \ell_k^2 \tan\left(\frac{\Theta_k}{2}\right)^2$ does not go above $\varepsilon^2$ threshold. For a given interval $I_k = [a_k,b_k]$ we do this as follows: 
\begin{enumerate}
    \item Compute $m_0:=p'(a_k)$, $m_1:=p'(b_k)$, $\ell_k^2$, $A=1+m_0m_1$, $B=(1+m_0^2)(1+m_1^2)$, $D=(m_0-m_1)^2$.
    \item If $A \leq 0$ this means turning angle is more than $\frac{\pi}{2}$, $I_k$ needs to be recomputed. Assume $A >0$, then 
    \[  \tan\left(\frac{\Theta_k}{2}\right)^2 = \frac{1-\cos(\theta_k)}{1+ \cos(\theta_k)} = \frac{\sqrt{B}-A}{\sqrt{B}+A} = \frac{B+A^2 - 2A\sqrt{B}}{D} \]
    \item Check if  $\ell_k^2(B+A^2) - 4 \varepsilon^2 D > 0$. If not, certification completed.
    \item Assuming $\ell_k^2(B+A^2) - 4 \varepsilon^2 D > 0$, check if
    \[  \left(\ell_k^2(B+A^2) - 4 \varepsilon^2 D\right)^2 <  4A^2B\ell_k^4 \]
\end{enumerate}

Note that for a fixed $a_k$ the error bound is a strictly increasing function of $b_k$. Therefore, we can easily turn this predicate into a greedy bisection algorithm with depth $O(\log(\frac{1}{\arctan(\varepsilon)}))$. 
Suppose we use this piecewise quadratic approximation for computing the CDTW distance, what is the magnitude of error in our computation?

\begin{lemma} \label{lemma:spacdtwerror}
Let $d_{CDTW}(P,Q)$ be the exact continuous dynamic time warping distance between curves $P$ and $Q$, and let $\tilde{d}_{CDTW}(P,Q)$ be the approximated distance computed using the piecewise quadratic curves $\tilde{P}$ and $\tilde{Q}$. If the approximation guarantees maximum spatial deviations of $\sup_x \norm{P(x) - \tilde{P}(x)}_2 \le \varepsilon$ and $\sup_y \norm{Q(y) - \tilde{Q}(y)}_2 \le \varepsilon$, and the total arc-lengths of the curves are $p_0$ and $q_0$, then we have:
\[ | \tilde{d}_{CDTW}(P,Q) - d_{CDTW}(P,Q) | \le 2\varepsilon \sqrt{p_0 + q_0}. \]
\end{lemma}

\begin{proof}
Let $\gamma = (\alpha, \beta) \in \Gamma(p_0, q_0)$ be any valid path. The exact cost evaluated along $\gamma$ is effectively the $L_2$ norm of the spatial difference vector with respect to the path-length measure $d\mu(z) = \norm{\gamma'(z)}_2 dz$:
\[ C(\gamma) = \left( \int_0^1 \norm{P(\alpha(z)) - Q(\beta(z))}_2^2 \norm{\gamma'(z)}_2 \, dz \right)^{\frac{1}{2}}. \]
Similarly, the approximated cost evaluated along the exact same path $\gamma$ is:
\[ \tilde{C}(\gamma) = \left( \int_0^1 \norm{\tilde{P}(\alpha(z)) - \tilde{Q}(\beta(z))}_2^2 \norm{\gamma'(z)}_2 \, dz \right)^{\frac{1}{2}}. \]

By Minkowski's inequality for $L_p$ spaces (specifically $L_2$), we can bound the exact cost by the approximated cost plus the norm of their difference:
\[ C(\gamma) \le \tilde{C}(\gamma) + \left( \int_0^1 \norm{ (P(\alpha(z)) - Q(\beta(z))) - (\tilde{P}(\alpha(z)) - \tilde{Q}(\beta(z))) }_2^2 \norm{\gamma'(z)}_2 \, dz \right)^{\frac{1}{2}}.\]

Applying the standard triangle inequality to the spatial deviation at any point $z$, we bound the error by the maximum deviations of our quadratic approximations:
\begin{align*}
\norm{ (P(\alpha(z)) - \tilde{P}(\alpha(z))) - (Q(\beta(z)) - \tilde{Q}(\beta(z))) }_2 &\le \norm{P(\alpha(z)) - \tilde{P}(\alpha(z))}_2 + \norm{Q(\beta(z)) - \tilde{Q}(\beta(z))}_2 \\
&\le 2\varepsilon.
\end{align*}

Substituting this constant upper bound into our integral gives:
\[ C(\gamma) \le \tilde{C}(\gamma) + 2\varepsilon \left( \int_0^1 \norm{\gamma'(z)}_2 \, dz \right)^{\frac{1}{2}}.\]

Because both $\alpha(z)$ and $\beta(z)$ are non-decreasing functions, we can bound the total $\ell_2$-arc-length of the path $\gamma$ in the parameter space by its $\ell_1$-arc-length:
\[ \int_0^1 \norm{\gamma'(z)}_2 \, dz = \int_0^1 \sqrt{\alpha'(z)^2 + \beta'(z)^2} \, dz \le \int_0^1 (\alpha'(z) + \beta'(z)) \, dz = p _0+ q_0. \]

Therefore, for any valid path $\gamma$, we establish the bound:
\[ C(\gamma) \le \tilde{C}(\gamma) + 2\varepsilon \sqrt{p_0 + q_0}. \]

Taking the infimum over all valid paths $\gamma \in \Gamma(p_0, q_0)$ on both sides yields the relation for the minimum costs:
\[ d_{CDTW}(P,Q) \le \tilde{d}_{CDTW}(P,Q) + 2\varepsilon \sqrt{p_0 + q_0}.\]

By applying the exact same argument symmetrically—reversing the roles of the exact curves $(P, Q)$ and the approximated curves $(\tilde{P}, \tilde{Q})$—we obtain the complementary bound:
\[ \tilde{d}_{CDTW}(P,Q) \le d_{CDTW}(P,Q) + 2\varepsilon \sqrt{p_0 + q_0}.\]

Combining these two inequalities yields the absolute value bound and completes the proof.
\end{proof}

We collect our results in this section into a single statement. Note that bit-size $\tau$ means the coefficients of the polynomials are of the order $2^{\tau}$.
\begin{prop} \label{lem:curve_approximation_error}
Suppose two algebraic curves of degree $d$ polynomials have bit-size $\tau$ and arc lengths $p_0$ and $q_0$. For a fixed $\varepsilon_0 >0$, we can subdivide them into $O(\frac{2^{\tau}d^2 \sqrt{p_0+q_0}}{\varepsilon_0})$  quadratic pieces using  $O(\tau + \log(\frac{d^2\sqrt{p_0+q_0}}{\varepsilon_0}) )$ bisections such that the resulting piecewise quadratic approximations  $\tilde{P}$ and $\tilde{Q}$ satisfies $| \tilde{d}_{CDTW}(P,Q) - d_{CDTW}(P,Q) | \le \varepsilon_0$. 
\end{prop}
\begin{proof}
The chord-length $\ell_k$ in Lemma \ref{lem:spatial_bound} are at most $d 2^{\tau}$ so picking $\varepsilon = \frac{\varepsilon_0}{2^{\tau} d \sqrt{p_0+q_0}}$ gives us a spatial distance bound of at most $\frac{\varepsilon_0}{\sqrt{p_0+q_0}}$. Using this inside \Cref{lemma:spacdtwerror} completes the proof of guarantee.

Note that for small $x > 0$, $\frac{x}{2} < \arctan(x) < x $. Using this for $x = \frac{\varepsilon_0}{2^{\tau} d \sqrt{p_0+q_0}}$ yields the estimates. 
\end{proof}

\begin{rem}
This result is stated and proved for two algebraic curves for clarity; however, it applies to piecewise algebraic case without any change.
\end{rem}
\section{Boundary Cost Functions and Cost Propagation} \label{propagation}
Our algorithm will be based on propagating a cost function from the input to output boundaries. The authors in \cite{buchin2022computing} were able to perform exact function propagation exploiting the fact that their focus is on piecewise linear curves. Our approach is inspired by their work; however, due to the non-linear nature of our setting, we cannot perform exact propagation. Ultimately, we subdivide the output boundary with control on approximation error and propagate a precise cost function on every subinterval in the output boundary. 

The theorem statement below concerns the case of an input cost function, a  polynomial of arbitrary degree $\delta_{\mathrm{in}}$, being propagated to the output boundary. It is stated and proved for cost propagation of arbitrary degree $d$ curves $P$ and $Q$. The purpose is to understand the structure of cost propagation in each cell. For our algorithms, we will first compute the piecewise quadratic approximation that was worked out in \Cref{turningangles}, and do the cost propagation on piecewise quadratics ($d=2$).

\begin{thm} \label{10}
    Let $f_{\mathrm{in}}(s)$ be a single cost subsegment of degree $\delta_{\mathrm{in}}$ and coefficient bit-size $\tau_{\mathrm{in}}$, defined over an input boundary segment of a cell. $f_{\mathrm{in}}(s)$ propagates to $O(1)$ continuous candidate cost branches along the output boundary, and all critical transition points partitioning these branches can be computed using exact algebraic real root isolation in $O(1)$ time.
\end{thm}

\begin{proof}
The proof will be divided into three pieces where every piece corresponds to one path type depicted in \Cref{threeoptions}. We start with the first case, the straight line segment.
\\\\
\textbf{Case 1: Type 1 Path}
Let the input boundary segment be linearly parameterized by $s \in [0, L_{\mathrm{in}}]$ with position $X(s) = X_0 + s \mathbf{e}_{\mathrm{in}}$, and the output boundary segment by $t \in [0, L_{\mathrm{out}}]$ with position $Y(t) = Y_0 + t \mathbf{e}_{\mathrm{out}}$, where $\|\mathbf{e}_{\mathrm{in}}\|_2 = \|\mathbf{e}_{\mathrm{out}}\|_2 = 1$. 
An optimal path $\gamma_{s,t}$ connecting input point $X(s)$ to output point $Y(t)$ is a straight line segment in parameter space $R$ described as \[ \gamma_{s,t}(z) = (1-z)X(s) + z Y(t), \] where $z \in [0,1]$.  Let $ \|\gamma_{s,t}'(z)\|_2 = \|Y(t) - X(s)\|_2 = \sqrt{D(s,t)}$,  where $D(s,t) := \|(Y_0 - X_0) + t \mathbf{e}_{\mathrm{out}} - s \mathbf{e}_{\mathrm{in}}\|_2^2$ is a bivariate quadratic polynomial in $s$ and $t$ with degree $\deg(D) = 2$. The integrated height along the normalized straight segment is given by:
\[
H(s,t) := \int_0^1 h\left((1-z)X(s) + z Y(t)\right) dz.
\]
Since $h(x,y)$ has degree $2d$ and coefficient bit-size $O(\tau)$, $H(s,t)$ is a bivariate polynomial of degree at most $2d$ in $s$ and degree at most $2d$ in $t$.

Given an input boundary cost function $f_{\mathrm{in}}(s)$ on $I = [0, L_{\mathrm{in}}]$ of degree $\delta_{\mathrm{in}}$ and coefficient bit-size $\tau_{\mathrm{in}}$, the propagated cost function at output parameter $t$ is:
\[
f_{\mathrm{out}}(t) = \inf_{s \in I} \left( f_{\mathrm{in}}(s) +\int_{0}^{1} h(\gamma_{s,t}(z))||\gamma'_{s,t}(z)||_2 \  dz \right)  = \inf_{s \in I} \left( f_{\mathrm{in}}(s) + \sqrt{D(s,t)} \cdot H(s,t) \right).
\]

For a fixed output parameter $t$, an optimal interior entry point $s^*(t)$ minimizes $G(s,t) := f_{\mathrm{in}}(s) + \sqrt{D(s,t)} H(s,t)$, satisfying $\frac{\partial G}{\partial s}(s,t) = 0$, or:
\[
f_{\mathrm{in}}'(s) + \frac{D_s(s,t) H(s,t) + 2 D(s,t) H_s(s,t)}{2 \sqrt{D(s,t)}} = 0,
\]
where $D_s = \frac{\partial D}{\partial s}$ and $H_s = \frac{\partial H}{\partial s}$. Multiplying by $2\sqrt{D(s,t)}$ and squaring both sides eliminates the radical, yielding the polynomial equation:
\begin{equation}
P(s,t) := \left[ D_s(s,t) H(s,t) + 2 D(s,t) H_s(s,t) \right]^2 - 4 \left( f_{\mathrm{in}}'(s) \right)^2 D(s,t) = 0.
\end{equation}
For fixed $t$, to find the roots $s$ of $P(s,t)=0$, we must run real root isolation algorithms. Note that the last squaring step does introduce extraneous roots; however, we assume these are checked and removed when roots are isolated.

To establish the time complexity of these root isolation algorithms, we bound both the degree $d_s := \deg_s(P)$ and the maximum coefficient bit-size $\tau_P$ of $P(s,t)$ as a polynomial in $s$. First we need two facts about bit-size bounds:
\begin{itemize}
    \item[(i)] \emph{(Product)} If $A(s), B(s)$ have bit-sizes $\tau_A, \tau_B$, then the bit-size of $A \cdot B$ is at most $\tau_A + \tau_B + \log_2(\min(d_A,d_B)+1)$ because the $k$-th coefficient of $A\cdot B$ is $\sum_{i+j=k} a_ib_j$, a sum of at most $\min(d_A,d_B)+1$ products of coefficients.
    \item[(ii)] \emph{(Sum)} If summed, $A + B$ has bit-size at most $\max(\tau_A,\tau_B) + 1$.
\end{itemize}
Now we proceed step-by-step through the functions, bounding degrees and bit-sizes:
\begin{enumerate}
    \item We show $\tau_H = O(\tau + d \log d)$ for $H(s,t)$. The height function $h(x,y) = \|P_i(x) - Q_j(y)\|_2^2$ essentially squares polynomials of degree $d$ and max coefficient bit-size $\tau$. This gives a bivariate polynomial of degree $2d$ with max coefficient bit-size $\tau_h = 2\tau + O(\log d)$.
    \\\\
    Next, along the straight path $\gamma_{s,t}(z) = ((1-z) \cdot x(s) + z\cdot x(t), \, (1-z)\cdot y(s) + z\cdot y(t))$, substituting the coordinates into $h(x,y)$ yields terms of the form $c_{a,b} \cdot ((1-z)s + z t_0)^a \cdot ((1-z)s_0 + z t)^b$ for $a+b \le 2d$. The  binomial expansion adds at most $2d = O(d)$ bits to the numerator of each coefficient. Integrating each expanded monomial $s^m t^n z^k (1-z)^{2d-k}$ with respect to $z$ over $[0,1]$ requires evaluating the Beta integral: \[ \int_0^1 z^k (1-z)^{2d-k} dz = B(k+1, 2d-k+1) = \frac{k! (2d-k)!}{(2d+1)!} = \frac{1}{(2d+1) \binom{2d}{k}}.\] Expressing the final polynomial $H(s,t)$ with exact integer coefficients requires clearing the denominators of all such Beta integrals across all terms of degree up to $2d$. The least common multiple (LCM) of these denominators is bounded by $(2d+1)!$. By Stirling's approximation, this requires $O(d \log d)$ bits to clear.  Summing the initial bit-size $\tau_h$, the $O(d)$ bits from binomial expansion, and the $O(d \log d)$ bits required to clear the denominators from the Beta integrals, the max coefficient bit-size of $H(s,t)$ is bounded by  $\tau_H = O(\tau + d \log d)$.
    \item We bound the bit-size of the partial derivatives $D_s$ and $H_s$:
    \begin{itemize}
        \item $D(s,t)$ is a bivariate quadratic polynomial ($\deg_s(D)=2$) with bit-size $\tau_D = O(\tau)$. The partial derivative $D_s(s,t) = \frac{\partial D}{\partial s}$ is linear in $s$. Differentiating multiplies coefficients by at most $2$, so $\tau_{D_s} = \tau_D + \log_2(2) = O(\tau)$.
        \item $H(s,t)$ has degree $2d$ in $s$ and bit-size $\tau_H = O(\tau + d \log d)$. The partial derivative $H_s(s,t) = \frac{\partial H}{\partial s}$ has degree $2d-1$ in $s$. Differentiating multiplies coefficients by at most $2d$, adding $\log_2(2d) = O(\log d)$ bits. Thus, $\tau_{H_s} = \tau_H + O(\log d) = O(\tau + d \log d)$.
    \end{itemize}

    \item We bound the bit-size of the products $D_s H$ and $H_sD$:
    \begin{itemize}
        \item $D_s H$ is the product of $D_s$ ($\deg_{s}=1, \tau_{D_s}=O(\tau)$) and $H$ ($\deg_{s}=2d, \tau_{H}=O(\tau + d \log d)$). Its degree in $s$ is $2d+1$, and its bit-size is bounded by $$\tau_{D_s} + \tau_H + \log_2(2) = O(\tau + d \log d).$$
        \item $D H_s$ is the product of $D$ ($\deg_{s} =2, \tau_{D} =O(\tau)$) and $H_s$ ($\deg_{s} =2d-1, \tau_{H_s} =O(\tau + d \log d)$). Its degree in $s$ is $2d+1$, and its bit-size is $$\tau_D + \tau_{H_s} + \log_2(3) = O(\tau + d \log d).$$
    \end{itemize}

    \item We bound the bit-size of $U(s,t) := D_s H + 2 D H_s$ and $U(s,t)^2$:
    \begin{itemize}
        \item The sum $U(s,t)$ has degree $2d+1$ in $s$. Adding two polynomials of bit-size $O(\tau + d \log d)$ adds at most $1$ bit, so $\tau_U = O(\tau + d \log d)$.
        \item Squaring $U(s,t)$ yields a polynomial of degree $2(2d+1) = 4d+2$ in $s$. The coefficients of $U^2$ contain at most $2d+2$ terms from $U$, so the maximum coefficient value is bounded by $(2d+2) \cdot (2^{\tau_U})^2$, giving a bit-size of \[\tau_{U^2} = 2\tau_U + \log_2(2d+2) = 2 \cdot O(\tau + d \log d) + O(\log d) = O(\tau + d \log d).\]
    \end{itemize}

    \item We bound the bit-size of $4 (f_{\mathrm{in}}'(s))^2 D(s,t)$. The piece $f_{\mathrm{in}}(s)$ has degree $\delta_{\mathrm{in}}$ and bit-size $\tau_{\mathrm{in}}$. Its derivative $f_{\mathrm{in}}'(s)$ has degree $\delta_{\mathrm{in}}-1$ and bit-size $\tau_{\mathrm{in}} + \log_2 \delta_{\mathrm{in}}$. Squaring $f_{\mathrm{in}}'(s)$ and multiplying it by $4$ gives degree $2\delta_{\mathrm{in}}-2$ and a bit-size $2\tau_{\mathrm{in}} + O(\log \delta_{\mathrm{in}})+\log_{2}{4}$. Finally, multiplying $4 (f_{\mathrm{in}}'(s))^2$ by $D(s,t)$ ($\deg_{s} =2, \tau_{D} =O(\tau)$) gives degree $2\delta_{\mathrm{in}}$ and bit-size $$\tau_{2} = 2\tau_{\mathrm{in}} + O(\tau + \log \delta_{\mathrm{in}}).$$

    \item We finally bound the degree and bit-size of $P(s,t)$. By subtracting the two terms $U(s,t)^2$ and $4 (f_{\mathrm{in}}'(s))^2 D(s,t)$, we obtain \[ d_s := \deg_s(P) = \max\left( 4d + 2, \, 2\delta_{\mathrm{in}} \right).\] The maximum coefficient bit-size $\tau_P$ of $P(s,t)$ is \[\tau_P = \max\left( \tau_{U^2}, \, \tau_{2} \right) + 1 = O\left( \tau + d \log d + \tau_{\mathrm{in}} + \log \delta_{\mathrm{in}} \right).\]
\end{enumerate}

So we have established that both the degree and the max coefficient bit-size of $P(s,t)$ are bounded by quantities depending only on $\tau,d,M, N$. By well-known real root isolation algorithms we can compute the real roots of $P(s,t)$, for fixed $t$, in $\tilde{O}_B\big(d_s^2\tau_P\big)$ time (see, for example, Section 1.2 of \cite{Ergur_2022} for a review of existing algorithms).

Now we would like to  use the techniques in \Cref{realag} to subdivide the output boundary to be able to propagate a single cost function in every sub-interval. We basically use the resultant based method of \Cref{realag} directly to compute the critical points; by construction on every sub-interval we have a fixed differentiable cost function. Moreover, finding these critical points takes constant time, as we have established for the real root isolation algorithms previously (because the max coefficient bit size of $P$ is controlled).\\
\textbf{Case 2: Type 2 Path}
Here the valley $Z(h)$ is the entire diagonal $x=y$. In this setup, an optimal path enters the diagonal valley $Z(h)$ at $v_{\mathrm{in}}$, traverses $Z(h)$ at zero cost, and exits at $v_{\mathrm{out}}$ toward $t$.

First, the entry cost $g_{\mathrm{entry}}(v_{\mathrm{in}})$ from input boundary $I$ to $v_{\mathrm{in}}$ is propagated via Case 1 straight line propagation in $O(1)$ time, yielding $O(1)$ polynomial pieces. 
Then, the accumulated cost along the diagonal is $g_Z(v) = \inf_{v_{\mathrm{in}} \le v} g_{\mathrm{entry}}(v_{\mathrm{in}})$, forming a cumulative minimum consisting of $g_{\mathrm{entry}}$ pieces and constant horizontal segments. Computing $g_Z(v)$ takes $O(1)$ time and produces $O(1)$ pieces.
Finally, propagation from $v_{\mathrm{out}}$ on the diagonal to output parameter $t$ is computed via Case 1 straight line propagation in $O(1)$ time, generating $O(1)$ output polynomial pieces.\\
\textbf{Case 3: Type 3 Path}
Here the valley $Z(h)$ consists of $k \le d$ discrete points $\{v_1, \dots, v_k\}$. For each $v_x$, the entry cost $g_{\mathrm{entry}}(v_x)$ from input boundary $I$ to $v_x$ is propagated via a simplified version of Case 1 straight line propagation with only one $t$ to consider. This takes $O(1)$ time. Then, minimum costs across discrete valley points are updated as $g_Z(v_1) = g_{\mathrm{entry}}(v_1)$ and $g_Z(v_x) = \min\left(g_{\mathrm{entry}}(v_x), \, g_Z(v_{x-1}) + C(v_{x-1}, v_x)\right)$ sequentially, from left to right. This takes $O(k) = O(d) = O(1)$ time. Finally, each $v_x$ projects to the output boundary via Case 1 straight line propagation, generating $O(1)$ polynomial pieces on the output boundary in $O(1)$  time.
\end{proof}

\section{Computing CDTW Distance of Piecewise Quadratic Curves} \label{errorandcost}
We will construct the main algorithmic engine in this section. The algorithm gives an approximation to CDTW distance of piecewise quadratic curves. Pseudo-code of the main algorithm and all sub-routines are in \Cref{pseudocode}. Recall that $P,Q$ have $m,n$ pieces respectively. We now use $M$ and $N$ to refer to the number of pieces of $\tilde{P}$ and $\tilde{Q}$, as these are the curves upon which we run our algorithm. Finally, we remind the reader that $\tau$ is the \textit{bit-size} of the coefficients of polynomials, and therefore the actual polynomial coefficients are on the order of $2^\tau$.

\subsection{Error and Coefficient Bit-Size Bounds}
\label{subsec:quintic}
 
Recall that \textproc{ApproxCDTW} (Algorithm~1, \Cref{pseudocode}) sweeps the parameter space $R$ cell by cell in order of increasing diagonal index $k=i+j$. At the start of diagonal step $k$, the algorithm holds a piecewise cost function on the diagonal boundary $A_{k-1}$ --- the union of the input boundaries of all cells $(i,j)$ with $i+j=k-1$ --- representing the (approximate) accumulated cost of an optimal path from the origin to each point of $A_{k-1}$. It propagates this cost function, one cell at a time, into the next diagonal boundary $A_k$ via \Cref{10}, and merges the resulting candidate branches on $A_k$ into a single piecewise cost function using a global lower envelope. Iterating this for $k = 2,\dots,M+N-2$ eventually reaches the top-rightmost corner $(p,q)$ of $R$ and returns an approximation of $\mathrm{CDTW}(P,Q)$.
 
What stops us from running this loop exactly is the mismatch between what \Cref{10} takes and what it produces. It requires the input piece $f_{\mathrm{in}}(s)$ to be a \textit{polynomial} of bounded degree $\delta_{\mathrm{in}}$ and bit-size $\tau_{\mathrm{in}}$; however, as the proof of \Cref{10} shows, propagating $f_{\mathrm{in}}$ through a cell produces $O(1)$ candidate branches of the form
\[
g(t) := G(s^*(t),t) = f_{\mathrm{in}}(s^*(t)) + \sqrt{D(s^*(t),t)}\, H(s^*(t),t),
\]
where $s^*(t)$ is a real root branch of the polynomial equation $P(s,t)=0$. These are \textit{algebraic}, or generally non-polynomial. For degree $\delta_{\mathrm{in}}=5$, the proof of \Cref{10} gives $\deg_s P = \max(4d+2,\,2\delta_{\mathrm{in}}) = 10$ for our quadratic curves ($d=2$). For a polynomial equation of degree $10$ there is in general no closed form solution, by Abel--Ruffini. Even besides a closed form, feeding $g(t)$ back into \Cref{10} as the next cell's $f_{\mathrm{in}}$ would compound the degree and bit-size of $P$ with every diagonal step, since \Cref{10}'s bounds on $\deg_s P$ and $\tau_P$ scale with the $\delta_{\mathrm{in}}$ and $\tau_{\mathrm{in}}$ of the input piece.
 
We resolve this by resetting every output branch $g(t)$ to a bounded-degree \textit{polynomial approximant} before it is used as an $f_{\mathrm{in}}$ at the next cell. Specifically, we replace $g$ by its degree-$5$ (quintic) Lagrange interpolant $p_{\mathrm{out}}$, at the cost of a small, controllable error $\sigma$ per piece. We choose degree $5$ because that is the degree produced by the base case of the recursion: $A_1.\mathrm{bottom}$ and $A_1.\mathrm{left}$ (Algorithm~1, lines 7--8) are antiderivatives of the quartic height function $h$ (degree $2d=4$ for our quadratic pieces $d=2$), and integration increases degree by one. Resetting every later branch to degree $5$ also keeps every input piece for \Cref{10} in the same bounded class at every cell of every diagonal, so that the propagation theorem holds throughout the whole sweep.
 
The main result of this section is the following:
\begin{thm} \label{thm:main-quintic}
Fix a global additive accuracy target $\eps_0 > 0$. Then, at every diagonal step $k = 2,\dots,M+N-2$ and for every cell on that diagonal, the propagation of \Cref{10} together with quintic approximation:
\begin{enumerate}
    \item produces $O\!\left(\left(\tfrac{M+N}{\eps_0}\right)^{1/6}\right)$ quintic output pieces per input piece per cell;
    \item introduces pointwise approximation error per piece, so that the total error accumulated along any propagated path over $O(M+N)$ diagonal steps is at most $\eps_0$;
    \item produces output coefficients of bit-size $\tau_{\mathrm{out}} = O\big(\tau+\log\tfrac{M+N}{\eps_0}\big)$ (crucially independent of $k$, avoiding an explosion);
    \item runs in time $\tilde{O}\big(\log\tfrac{M+N}{\eps_0}\big)$ per piece.
\end{enumerate}
\end{thm}
As a corollary, we obtain an additive  error bound as follows.
\begin{cor}[Additive Error Accumulation] \label{cor:accum}
Let $E_k$ denote the sup-norm error, relative to the true optimal cost, of the piecewise cost function on diagonal $A_k$. Then we have $E_{M+N-2} \le \eps_0/2$.
\end{cor}

The remainder of this section is carried out in three steps: (i) we bound the error incurred by replacing $g$ with $p_{\mathrm{out}}$ on a single sub-interval of the output boundary (\Cref{sssec:regime1error}); (ii) we bound the coefficient bit-size of $p_{\mathrm{out}}$ by $O(1)$ \Cref{thm:coeffbitsize}; and (iii) we show this bit-size bound is preserved from one diagonal step to the next, so it does not compound (\Cref{sssec:nocompound}, Proposition \ref{prop:nocompound}).

\subsubsection{Error Bounds for The Quintic Approximation}
\label{sssec:regime1error}

Fix a sub-interval $[t_a,t_b]$ of the output boundary produced by \Cref{10} (Case 1, or Cases 2--3 via the reduction to Case 1 noted above), on which a single branch $s^*(t)$ of $P(s,t)=0$ is the optimal minimizer. Let $L := t_b-t_a$, $t_0 := \tfrac{t_a+t_b}{2}$, and
\[
g(t)  :=  f_{\mathrm{in}}(s^*(t)) + \sqrt{D(s^*(t),t)}\,H(s^*(t),t).
\]
 
\begin{dfn} [Critical Points] \label{dfn:regime}
We define the discriminant like polynomial using resultants as follows:
\[
\mathcal{C}(t)  :=  \mathrm{Res}_s\!\left(P(s,t),\, \frac{\partial P}{\partial s}(s,t)\right).
\]
\end{dfn}

By well-known degree and height bounds for resultants of bivariate polynomials, $\mathcal{C}(t)$ is a univariate polynomial of degree $d_C = O(1)$ and bit-size $\tau_C = O(\tau_P)$. 
Consider the roots of $\mathcal{C}(t)$ in two groups: double-roots (singularity) and isolated roots. If there are any double-roots, we process them as follows:  The two colliding branches have a local Puiseux expansion:
\[
s^*(t)  =  s_0 \pm c\sqrt{t-t_a} + O(t-t_a) \Rightarrow  g(t)  =  g(t_a) + c'\sqrt{t - t_a} + O(t-t_a),
\]
with leading coefficient $|c'| \le 2^{O(\tau_P)}$, computable with methods such as Newton polygons \cite{PuiseuxExpansionComputation}. Since a degree-$5$ polynomial cannot approximate a square-root function on an interval $[0,T]$ to better than $ O(\sqrt{T})$ uniformly, regardless of node placement, we excise an interval of half-width
\[ \ell_0  :=   O\!\left(\sigma^2 / |c'|^2\right)  =  2^{-O(\tau_P)} \sigma^2\]
adjacent to each singularity, created via $O(\tau_P + \log(1/\sigma))$ further bisections, and replace $g$ on the buffer by the constant $g(t_a)$, which is accurate to $O(\sigma)$. This adds $O(1)$ extra pieces per singularity (still $O(1)$ total per sub-interval, since $\mathcal{C}$ has $O(1)$ roots by \Cref{dfn:regime}) and does not change the bisection depth of Proposition \ref{prop:regularerror}.

Now assume the double-roots are processed as described, or there are no double-roots. Then, for the remaining sub-intervals $s^*(t)$ is single-valued and real-analytic, thus extending to a holomorphic function of $t$ on a complex disk around $t_0$ whose radius is determined by the nearest complex root of $\mathcal{C}$ (by the holomorphic implicit function theorem). Thus, in Regime 1, classical polynomial interpolation converges. The error and bit-size bounds can be calculated as shown in \Cref{sssec:regime1error}. In particular, we rely on the classical Davenport-Mahler-Mignotte root separation bound.

By the classical Davenport-Mahler-Mignotte Root Separation Bound \cite{Sharma_2020}, any two distinct roots of $\mathcal{C}$ (which recall has constant degree and bit-size $\tau_C=O(\tau_P)$) are separated by at least $\rho := 2^{- O(\tau_P)}$. Thus, there is a certain radius around any root where no other roots are present. This, called the holomorphy radius, is given as
\[
r := \tfrac{1}{2}\rho = 2^{- O(\tau_P)},
\]
which clearly only depends on the bit-size $\tau_P$ of the cell's $P(s,t)$ polynomial.
 
We now bound the maximum magnitude of $g$ as follows.
\begin{lemma}[Magnitude Bound via Accumulated Cost] \label{lem:magnitude}
For all $t$ in the domain, $0 \le g(t) \le M_g$, where
\[
M_g  :=  2^{O(\tau)} \cdot (M+N).
\]
\end{lemma}
\begin{proof}
The value $g(t)$ is (an approximation to) the accumulated cost of an optimal path from the origin $(0,0)$ of the parameter space to a point on the current cell boundary. It is a sum of at most $M+N$ per-cell integrals $\int h(\gamma(z)) \|\gamma'(z)\|_2\,dz$, each taken over a curve piece of coefficient bit-size $O(\tau)$ and bounded parameter length. Since $h$ is quartic and has bit-size $O(\tau)$ (see proof of \Cref{10}), and each piece's arc length is $2^{O(\tau)}$-bounded, each cell's contribution to the accumulated cost is bounded by $2^{O(\tau)}$. 
\end{proof}
 
Now we use the standard Chebyshev interpolation formulation to bound the error between $g(t)$ and its quintic interpolant approximation $p_{\mathrm{out}}$. For completeness, we give a proof.
\begin{lemma}[Error from Degree-5 Chebyshev Interpolation] \label{lem:cauchyremainder}
Let $g$ be holomorphic and bounded by $M_g$ on the closed disk $\overline{\Delta_r(t_0)}$, with $r \ge L$, and let $p_{\mathrm{out}}$ be the degree-$5$ interpolant of $g$ at the $6$ Chebyshev nodes of $[t_a,t_b]$. Then we claim:
\[
\sup_{t \in [t_a,t_b]} |g(t) - p_{\mathrm{out}}(t)|  \le  4 M_g \left(\frac{L}{2r}\right)^6.
\]
\end{lemma}
\begin{proof}
Writing $\omega(u) = \prod_{k=0}^5(u-t_k)$ for the nodal polynomial, Hermite's contour-integral remainder formula gives, for $t \in [t_a,t_b]$ inside $\Delta_r(t_0)$, that:
\[
g(t) - p_{\mathrm{out}}(t)  =  \frac{1}{2\pi i} \oint_{|z-t_0|=r} \frac{\omega(t)}{\omega(z)} \cdot \frac{g(z)}{z-t}\,dz.
\]
For the Chebyshev nodes rescaled to $[t_a,t_b]$, it is known that $\max_{[t_a,t_b]} |\omega| = 2(L/4)^6$. For $|z-t_0| = r \ge L$, we have $|\omega(z)| \ge (r-L/2)^6 \ge (r/2)^6$ and $|z-t| \ge r - L/2 \ge r/2$. Substituting into the contour integral and bounding $|g(z)| \le M_g$ gives
\[
|g(t) - p_{\mathrm{out}}(t)|  \le  \frac{2(L/4)^6}{(r/2)^6} \cdot \frac{r}{r/2} \cdot M_g  =  4 M_g \left(\frac{L}{2r}\right)^6.
\]
\end{proof}

\begin{prop} \label{prop:regularerror}
We can subdivide an interval further into sub-intervals of length $L =  O\big(r\,\sigma^{1/6}/M_g^{1/6}\big)$, via $\beta = O\big(\tau_P + \log(1/\sigma) + \log(M+N)\big)$ bisections. This guarantees
\[
\sup_{[t_a,t_b]} |g - p_{\mathrm{out}}|  \le  \sigma.
\]
\end{prop}
\begin{proof}
Substituting $r = 2^{- O(\tau_P)}$ and $M_g = 2^{O(\tau+ \log(M+N))}$ (by Lemma \ref{lem:magnitude}) into Lemma \ref{lem:cauchyremainder} and solving $4M_g(L/2r)^6 \le \sigma$ for $L$ gives the claimed length. Bisection of a dyadic interval is exact and halves its length at each step. Since the original Regime 1 intervals are of length bounded by $2^{O(\tau)}$, the depth required to further subdivide them into intervals of length $L$ is $$\beta = \log_2(\text{initial length}/L) = O\big(\tau_P + \log(1/\sigma) + \log(M+N)\big).$$
\end{proof}

\begin{thm} \label{thm:quinticerror}
At every diagonal step and every candidate output boundary branch, there is a quintic polynomial $p_{\mathrm{out}}$, computable via $O\big(\tau_P + \log(1/\sigma) + \log(M+N)\big)$ bisections and $O(1)$ node evaluations, such that
\[
\sup_t |g(t) - p_{\mathrm{out}}(t)|  \le  \sigma
\]
on every resulting sub-interval, with $O(\sigma^{-1/6})=O\!\left(\left(\tfrac{M+N}{\eps_0}\right)^{1/6}\right)$ sub-intervals produced per output branch.
\end{thm}
\begin{proof}
This immediately follows from Proposition \ref{prop:regularerror} and the discussion of processing double-roots of the resultant polynomial. While $\mathcal{C}$ has $O(1)$ roots, the bisection into intervals of length $L = O(r\sigma^{1/6}/M_g^{1/6})$ partitions $[0,L_{\mathrm{out}}]$ into $O(1/L) = O(\sigma^{-1/6})$ pieces. Then, setting $\sigma = \frac{\eps_0}{M+N}$ gives $O(\sigma^{-1/6}) = O\!\left(\left(\tfrac{M+N}{\eps_0}\right)^{1/6}\right)$ pieces, as claimed.
\end{proof}

\subsubsection{Coefficient Bit-Size Bound for the Quintic}
\label{sssec:bitsize}
Now we bound the coefficient bit-size of the interpolant $p_{\mathrm{out}}$, so that it is a valid input piece $f_{\mathrm{in}}$ the next time \Cref{10} is applied, or in the cells of the following diagonal. In particular, we show that the interpolant retains a coefficient bit-size bounded by a constant.

First we consider the bit-size of the Chebyshev nodes used in the interpolation as follows.
\begin{lemma}[Node Bit-Size] \label{lem:node-bitsize}
Every dyadic node $t_k$ produced by the bisection of Proposition \ref{prop:regularerror}  satisfies $t_k = m_k/2^{B_t}$ with
\[
B_t  =  O\big(\tau + \tau_P + \log(1/\sigma) + \log(M+N)\big).
\]
\end{lemma}
\begin{proof}
Bisection of a dyadic interval adds exactly one bit to the denominator exponent per step, so, starting from an $O(\tau)$-bit-length cell boundary, performing $\beta = O\big(\tau_P+\log(1/\sigma)+\log(M+N)\big)$ steps (see Proposition \ref{prop:regularerror}) gives the claimed bound.
\end{proof}

We must then refine the isolating intervals using quadratic interval refinement (see \cite{kerber2015QIR}), so that the evaluation of the interpolant stays within the required accuracy target. The complexity of this refinement is bounded as follows.

\begin{thm}[Refinement Complexity] \label{thm:refinement}
Given an isolating interval for $s^*(t_k)$, refining it (and consequently the evaluation of $y_k := g(t_k)$) to $b$ bits costs $\tilde{O}_B(d_s^3 \tau_P^2 + d_s b)$ bit operations, using quadratic interval refinement.
\end{thm}

Moreover, by standard results \cite{Ibrahimoglu2016Lebesgue}, we have that rounding sample values to an absolute error $\eta_{\mathrm{err}}$ perturbs the quintic Chebyshev interpolant by at most $3\,\eta_{\mathrm{err}}$ (in supremum norm). Thus, so that the rounded value $\tilde{y}_k$ has absolute error $\le \sigma/3$, we take $b = O\big(\tau+\log(1/\sigma)+\log(M+N)\big)$. This gives us a sample bit-size
\[
B_y  :=  O\big(\tau+\log(1/\sigma)+\log(M+N)\big).
\]
Furthermore, let $B^* := B_t + B_y = O\big(\tau+\tau_P+\log(1/\sigma)+\log(M+N)\big)$.
 
We establish the bound for the coefficient bit-size of the interpolant as follows.
\begin{thm}[Coefficient Bit-Size Bound of Quintic Interpolating Polynomial] \label{thm:coeffbitsize}
The coefficients of $p_{\mathrm{out}}$, expressed as dyadic rationals, have bit-size
\[
\tau_{\mathrm{out}}  =  O(B^*) = O\big(\tau+\tau_P+\log(1/\sigma)+\log(M+N)\big).
\]
\end{thm}
\begin{proof}
Each Lagrange basis polynomial $\ell_k(t) = \prod_{j \ne k} (t-t_j)/(t_k-t_j)$ comes from the $6$ fixed nodes $t_j$ (with bit-size $B_t$) combined via a constant number of arithmetic operations: $5$ subtractions and products forming the numerator, and $4$ products forming the denominator constant. Since pairwise node separation is bounded by $\Theta(L)$ (where $L = 2^{- O(B_t)}$ by construction), the denominator satisfies $$2^{-O(B_t)} \le |\prod_{j\ne k}(t_k-t_j)| \le 2^{O(B_t)},$$ so its reciprocal has bit-size $O(B_t)$. Combining a fixed number of sums and products of $O(B_t)$-bit dyadic rationals yields the coefficient size bound on $\ell_k$ as $O(B_t)$.
 
Finally, multiplying by the rounded sample $\tilde{y}_k$ (with bit-size $B_y$) and summing the $6$ terms adds the bit-sizes: $O(B_t) + O(B_y) + O(1) = O(B^*)$, as claimed.
\end{proof}
 
$\tau_{\mathrm{out}}$ depends on characteristics of the current cell and problem setup, and so does not grow as propagation progresses; thus, $\tau_{\mathrm{out}} = O(1)$, allowing the real root isolation algorithms to run in constant time in \Cref{10}. Note that $\tau_{\mathrm{in}}$ only shows up in the bound for $\tau_P$ (see proof of \Cref{10}), which only affects the number of bisections performed (affecting the running time) and not $\tau_{\mathrm{out}}$ because we always reconstruct $p_{\mathrm{out}}$ from rounded interpolating data.
 
\subsubsection{No Bit-Size Explosion Across Diagonal Steps}
\label{sssec:nocompound}
The bound of \Cref{thm:coeffbitsize} is stated for a single cell, in terms of the bit-size $\tau_{\mathrm{in}}$ of that cell's own input piece. However, $\tau_{\mathrm{in}}$ is itself the output bit-size $\tau_{\mathrm{out}}$ of a piece produced one diagonal earlier. For the bound to be of any use across the whole sweep, it must therefore hold with the same constant at every diagonal $k$, rather than growing with $k$. We show precisely this, that the constant bound on coefficient bit-size is retained at each inductive step, from diagonal to diagonal.

\begin{prop}[Bit-Size Bound Across Diagonal Steps] \label{prop:nocompound}
There is a constant $C$, independent of $M, N, \eps_0$, and the diagonal index $k$, such that if every quintic output piece produced at diagonal step $k-1$ has bit-size at most
\[
B^*(\sigma)  :=  C\Big(\tau+\log(1/\sigma)+\log(M+N)\Big),
\]
then every quintic output piece produced at diagonal step $k$ also has bit-size at most $B^*(\sigma)$.
\end{prop}
\begin{proof}
We proceed with induction on $k$. The base case $k=1$ holds because $A_1.\mathrm{bottom}$ and $A_1.\mathrm{left}$ are exact integrals of $h$ along the parametrized boundary, resulting in bit-size $O(\tau) \le B^*(\sigma)$ for suitable $C$.
 
For the inductive step, suppose $\tau_{\mathrm{in}} \le B^*(\sigma)$ for every piece on $A_{k-1}$. Then we can substitute $\tau_P = O(\tau+\tau_{\mathrm{in}}+\log\delta_{\mathrm{in}}) = O\big(B^*(\sigma)\big)$ (see proof of \Cref{10}) into \Cref{thm:coeffbitsize} gives
\[
\tau_{\mathrm{out}}  =  O\Big(\tau+O(B^*(\sigma))+\log(1/\sigma)+\log(M+N)\Big),
\]
again \textit{independent of} $k$. Choosing $C$ large enough to absorb the fixed chain of constants finishes the induction.
\end{proof}

\subsection{Bounding the Complexity}
We now bound the total computational complexity of our algorithm \textproc{ApproxCDTW}. Let $T$ be the total number of boundary cost function pieces over all boundaries of all cells. Let $T_{\mathrm{env}}$ denote the time spent calculating the global lower envelopes. It follows that the overall complexity of this algorithm is $O(T + T_{\mathrm{env}})$. To bound this complexity further, we first provide a few definitions.

Define $A_k$ to be the union of the input boundaries of the cells $(i,j)$ such that $i+j=k$. Alternatively, $A_k$ is the union of the output boundaries of the cells $(i,j)$ such that $i+j=k-1$. Next, construct the partition 
\[
A_k \coloneqq \{ A_{k,1}, A_{k, 2}, \dots, A_{k, |A_k|} \}
\]
of $A_k$ into subsegments, where the subsegment $A_{k,\ell}$ is the segment between the $\ell^{\mathrm{th}}$ and $(\ell+1)^{\mathrm{th}}$ critical point along $A_k$. We define a critical point along $A_k$ as any of the following:
\begin{enumerate}
    \item A cell corner along $A_k$,
    \item Where the valley  meets the boundary $A_k$, or
    \item Where an optimal path to $A_k$ switches between two subintervals of $A_{k-1}$.
\end{enumerate}
Define $|A_{k,\ell}|$ to be the number of polynomial pieces in the piecewise cost function along the subsegment $A_{k,\ell}$. We express the total number of boundary polynomial pieces $T$ as:
\[
T = \sum_{k=2}^{M+N-2} \sum_{\ell=1}^{|A_k|} |A_{k,\ell}|.
\]
We first show, for all $k$, that the number of subsegments $|A_k|$ is bounded by $O((M+N)^2)$; then, we show that $|A_{k,l}|$ is bounded by $O\!\left(\left(\tfrac{M+N}{\eps_0}\right)^{1/6}\right)$ for any $k,l$.

Consider the propagation of cost functions from $A_{k-1}$ to $A_k$. We want to bound how the number of subintervals $|A_k|$ grows relative to $|A_{k-1}|$. Subintervals are separated by critical points, which can be cell corners, valley intersections with the output boundary, or points where the optimal path switches between source subintervals of $A_{k-1}$. Firstly, by \Cref{path_spatiality}, optimal paths with distinct starting points do not cross. Then, along $A_k$, the mapping back to optimal input locations on $A_{k-1}$ must be monotonic. An optimal path to $A_k$ can switch between two given subintervals of $A_{k-1}$ at most once, because alternating back and forth would require the corresponding optimal paths to intersect, contradicting \Cref{path_spatiality}. This contributes $|A_{k-1}|$ subintervals at most to $|A_k|$. Secondly, the total $2k$ cell corners and valley intersections with the output boundary give $2k$ subintervals.  So, we have $|A_k| \leq |A_{k-1}| + 4k$.  This immediately yields $|A_{k}| \leq 3k^2 = 3(k-1)^2 + 6k - 3$ with a straight-forward induction. Thus, we have  $A_{M+N}=O((M+N)^2)$. 

Each candidate branch $G_i(t)$ on $A_k$ satisfies a minimizer equation $P(s,t) = 0$ of degree $d_s, d_t = O(1)$ (see \Cref{10}). Then, for any two branches $G_a(t)$ and $G_b(t)$, there are at most $\kappa = O(d_s) = O(1)$ real roots on $A_k$. Thus every pair of candidate branches intersects at most $\kappa = O(1)$ times.

To prevent exponential path growth, we construct the global lower envelope $\mathcal{E}_k(t) := \min_{1 \le i \le m_k} G_i(t)$ along $A_k$. Recall that an $(n,s)$ Davenport--Schinzel sequence is a sequence over $n$ symbols with no two adjacent symbols equal and no alternating subsequence $a \cdots b \cdots a \cdots b \cdots$ of length $s+2$ between any two distinct symbols $a,b$. Let $\lambda_s(n)$ denote the maximum length of such a sequence. If $\mathcal{F}$ is a family of $n$ continuous, partially defined functions such that every pair intersects in at most $s$ points, then the left-to-right sequence recording which function of $\mathcal{F}$ attains the minimum is itself an order-$(s+2)$ Davenport--Schinzel sequence, so the lower envelope of $\mathcal{F}$ has at most $\lambda_{s+2}(n)$ pieces \cite[Ch.~1]{sharir1995davenport}. Applying this with $\mathcal{F}$ as the $m_k$ candidate branches on $A_k$ and $s = \kappa = O(1)$, we get $|A_k| \le \lambda_{\kappa+2}(m_k)$.

It remains to bound $\lambda_{\kappa+2}(m_k)$. Known bounds give $\lambda_s(n) = O\left(n \cdot 2^{\alpha(n)^{c_s}}\right)$ for a constant $c_s$ depending only on $s$ \cite[Ch.~1--3]{sharir1995davenport}, where $\alpha(\cdot)$ is the inverse Ackermann function. Since $\alpha(\cdot)$ is at most $5$ for any input of conceivable practical size, we treat $2^{\alpha(m_k)^{c_s}}$ as constant for our fixed $s = \kappa+2 = O(1)$, so $\lambda_{\kappa+2}(m_k) = O(m_k)$. Finally, plugging in $m_k = O\left((M+N)^2\right)$, we obtain
\begin{equation}
|A_k| \le \lambda_{\kappa+2}(m_k) = O(m_k) = O\left( (M+N)^{2} \right). \label{eq:ds_bound}
\end{equation}

We now bound $|A_{k,\ell}|$ for any subsegment $A_{k,\ell}$ of the lower envelope $\mathcal{E}_k(t)$.

\begin{claim} \label{7.11}
For each subsegment $A_{k,\ell}$ along $A_k$,
\begin{equation}
|A_{k,\ell}| \le 2 \cdot  N(A_{k,\ell}) + 2 = O\!\left(\left(\tfrac{M+N}{\eps_0}\right)^{1/6}\right), \label{666}
\end{equation}
where $N(A_{k,\ell})$ counts the number of distinct quintic coefficient tuples that the candidate branch of $A_{k,\ell}$ passes through.
\end{claim}

\begin{proof}
By definition, $A_{k,\ell}$ contains no internal lower envelope branch intersections, so the cost function along $A_{k,\ell}$ is controlled by a single smooth candidate branch $G_i(t)$ (or a constant valley segment from Case 2). By \Cref{thm:quinticerror}, $G_i(t)$ consists of $O\!\left(\left(\tfrac{M+N}{\eps_0}\right)^{1/6}\right)$ quintic sub-intervals over its entire domain, so $N(A_{k,\ell})$ is also at most this count of $O\!\left(\left(\tfrac{M+N}{\eps_0}\right)^{1/6}\right)$. Taking a cumulative minimum along valley paths inserts horizontal constant functions only at subsegment endpoints or local minima; since each coefficient tuple is a single quintic polynomial piece, its derivative has at most $4$ real roots, contributing at most $2$ local minima. Since there are at most $2$ endpoints and at most $2$ local minima per coefficient tuple, the cumulative minimum adds at most $2 \cdot  N(A_{k,\ell}) + 2 = O\!\left(\left(\tfrac{M+N}{\eps_0}\right)^{1/6}\right)$ pieces in $A_{k,\ell}$.
\end{proof}

Substituting the bounds \eqref{eq:ds_bound} and \eqref{666} into $T$ yields:
\begin{align*}
    T &= \sum_{k=2}^{M+N-2} \sum_{\ell=1}^{|A_k|} |A_{k,\ell}| \\ &\le \sum_{k=2}^{M+N-2} |A_k| \cdot O\!\left(\left(\tfrac{M+N}{\eps_0}\right)^{1/6}\right) \\ &\le \sum_{k=2}^{M+N-2} O\left((M+N)^{2 }\left(\tfrac{M+N}{\eps_0}\right)^{1/6}\right) \\ &= O\left((M+N)^{3 }\left(\tfrac{M+N}{\eps_0}\right)^{1/6}\right).
\end{align*}

We now find the envelope construction complexity $T_{\mathrm{env}}$. At each diagonal $A_k$, candidate curves $G_i(t)$ arrive in a monotonic spatial ordering along $A_k$ due to the non-crossing path property (\Cref{path_spatiality}). Consequently, constructing the lower envelope over $m_k = O\left((M+N)^2\right)$ spatially ordered candidate curves takes linear time $O(m_k) = O\left((M+N)^2\right)$ via a single-pass stack-based sweep (see Algorithm 6). Summing over all $M+N-2$ diagonals yields:
\[
T_{\mathrm{env}} = \sum_{k=2}^{M+N-2} O\left((M+N)^2\right) = O\left((M+N)^3\right).
\]
Since $\eps_0$ is a target accuracy, we take $\eps_0 \le 1$, so $\eps_0^{-1/6} \ge 1$ and the $T_{\mathrm{env}}$ term is dominated by $T$. Combining $T$ and $T_{\mathrm{env}}$, we obtain the total of
\[
O(T + T_{\mathrm{env}}) = O\left((M+N)^{3 }\left(\tfrac{M+N}{\eps_0}\right)^{1/6}\right)
\]
\textit{piece operations}.

Finally, we must account for the quintic interpolation. By \Cref{thm:main-quintic}, fixing a global accuracy target $\eps_0$, each of the $T = O((M+N)^{3 })$ boundary polynomial pieces now costs $\tilde{O}\big(\log\tfrac{M+N}{\eps_0}\big)$ bit operations to be produced, and rounded to a bounded-bit-size quintic. The resulting algorithm computes $\mathrm{CDTW}(P,Q)$ exactly within additive error $\eps_0$. The overall bit-complexity of \textproc{ApproxCDTW} is therefore
\begin{equation} \label{complexity}
O(T+T_{\mathrm{env}}) \cdot \tilde{O}\Big(\log\frac{M+N}{\eps_0}\Big)  =  O\left(\frac{(M+N)^{19/6}}{\eps_0^{1/6}} \log\frac{M+N}{\eps_0}\right).    
\end{equation}

\section{An FPTAS for CDTW Distance of Piecewise Algebraic Curves} \label{FPTAS}
This section presents the final construct in a simple way. The pieces are already worked out in detail and the construct is likely clear to the reader. We present it briefly for clarity. We are given two piecewise algebraic curves $P$ and $Q$ with $m$ and $n$ pieces and total-arc length $p$ and $q$. The algebraic pieces are defined by degree $d$ polynomials of bit-size at most $\tau$. We do the following:
\begin{enumerate}
    \item Compute the Fr\'{e}chet distance between $P$ and $Q$ in $O(mn \log(mn))$ time using \cite{rote2007computing}. Derive the lower bound to CDTW distance of $P$ and $Q$ using the main theorem of  \Cref{lower-bound}. Call this lower bound $\delta$. Note that $\delta$ has no dependency on $\varepsilon$: it is a lower bound produced by the mathematical result in \Cref{lower-bound} and exact algorithms for Fr\'{e}chet distance computation. The purpose of $\delta$ is to bridge additive error bounds to multiplicative error bounds on CDTW distance.
    \item Pick the multiplicative approximation target $\varepsilon$. Set the additive error bound $\varepsilon_0 = \delta \varepsilon$.
    \item Subdivide both $P$ and $Q$ using turning angles and create quadratic Bezier approximations. Here we guarantee the CDTW approximation error in Lemma \ref{lem:curve_approximation_error} is at most $\frac{\varepsilon_0}{2}$. This yields piecewise quadratic approximations to each algebraic piece, where each piecewise quadratic approximation has $O(\frac{d^2 2^{\tau} \sqrt{p+q}}{\varepsilon_0})$ many pieces, and takes $O(\tau + \log(\frac{d^2 \sqrt{p+q}}{\varepsilon_0}) )$ bits to create.
    \item Approximate the CDTW distance of piecewise quadratic curves where we guarantee the total approximation error is at most $\frac{\varepsilon_0}{2}$.  Now we have in total $M= O(\frac{m d^2 2^{\tau} \sqrt{p+q}}{\varepsilon_0})$ and $N=O(\frac{n d^2 2^{\tau} \sqrt{p+q}}{\varepsilon_0})$ quadratic pieces in our approximation. Using \Cref{complexity}, we conclude that this approximation can be done with 
    \[ O \left(  \left((m+n) d^2 2^{\tau} \sqrt{p+q} \right)^{\frac{19}{6}} (\frac{1}{\varepsilon_0})^{\frac{10}{3}} \log \left( \frac{ (m+n) d^2 2^{\tau} \sqrt{p+q}}{\varepsilon_0^2} \right) \right)  \]
    bit operations. The parameters $\tau$, arc-length, and the lower bound, are constant that do not depend on $\varepsilon$. So, if we ignore them to clarify the dependency on $\varepsilon, m,n$, the algorithm takes
        \[ O \left(  \left( (m+n) d^2\right)^{\frac{19}{6}} \left(\frac{1}{\varepsilon}\right)^{\frac{10}{3}} \log \left( \frac{ (m+n) d}{\varepsilon^2} \right) \right)  \]    
        bit operations.
\end{enumerate}
\begin{rem}
We remark that the analysis of the algorithm for Fr\'{e}chet distance --- \cite{rote2007computing} and the recent elegant variant focused on the decision-version of the problem \cite{Conradi_Driemel_Kolbe_2025} --- does not keep track of the dependency on $d$, treating it as a constant. We did not attempt to rigorously extract the dependency on $d$ from these works, but the dependency is likely sub-quadratic. Therefore the dependency on $d$ in our final estimate is likely missing a $d^2$ term.
\end{rem}
\section{Conclusion}
In this paper, we presented an algorithmic framework for computing the Continuous Dynamic Time Warping (CDTW) distance between piecewise algebraic curves. Due to the algebraic hardness of computing exact CDTW distances under the squared Euclidean $\ell_2$ norm, our framework  approximates the CDTW distance to any user-specified error $\varepsilon$.

We proved via Pontryagin's maximum principle that optimal alignments consist of straight line segments or segments along valleys of the height function. We also obtain new results on comparing CDTW distance with Fr\'{e}chet distance for piecewise smooth curves: this allows us to create a fast and coarse approximation to CDTW distance by computing Fr\'{e}chet distance.  By approximating arbitrary algebraic curves with piecewise quadratic curves, and propagating optimal costs through cell boundaries using quintic Chebyshev interpolation and Davenport-Schinzel lower envelopes, we strictly bounded both the computational complexity and accumulated errors. 

To the best of our knowledge, this is the first algorithm for approximating CDTW distance of piecewise smooth curves. We believe the mathematical results and approximation machinery we developed will remain useful for future work in this field.  Our work is a first step toward extending CDTW algorithms beyond piecewise linear curves, and we hope that future work will obtain much faster algorithms applicable to an even broader family of curves.

\section{Acknowledgements}
We would like to thank Kevin Buchin and Samson Wang for answering our questions about their nice work \cite{buchin2022computing}, and to Jonathan de Koning for helpful discussions and meticulous experimentation in the early phases of this project. A.E. is also grateful to Claire Walton for her wonderful lectures on optimal control, and to the support from NCF CCF 2414160.

\printbibliography

\appendix
\section{Pseudo-Code of the Main Algorithms} \label{pseudocode}
We provide the pseudo-code of our algorithm \textproc{ApproxCDTW} (Algorithm 1), along with supporting sub-algorithms (Algorithms 2 through 6), that compute the optimal CDTW distance via diagonal-to-diagonal propagation of individual cells, global lower envelope reduction, and quintic interpolant approximations.

\begin{algorithm}[H]
\caption{ApproxCDTW}
\begin{algorithmic}[1]
\Require Piecewise curves $P$ and $Q$ split into $M$ and $N$ quadratic pieces, respectively; target additive accuracy $\eps_0 > 0$.
\Procedure{ApproxCDTW}{$P, Q, \eps_0$}
    \State $M \gets \text{number of segments in } P$
    \State $N \gets \text{number of segments in } Q$
    \State $p \gets L_2 \text{ arc length of } P$
    \State $q \gets L_2 \text{ arc length of } Q$
    \State $\sigma \gets \eps_0 / (2(M+N))$
    
    \State $A_1.\text{bottom}(x) \gets \int_0^x \|P(z) - Q(0)\|_2^2 \, dz \quad \forall x \in [0, p]$ \Comment{Initialize base boundary $A_1$}
    \State $A_1.\text{left}(y) \gets \int_0^y \|P(0) - Q(z)\|_2^2 \, dz \quad \forall y \in [0, q]$
    
    \For{$k = 2$ \textbf{to} $M + N - 2$} \Comment{Loop over diagonal wavefronts $A_k$}
        \State $\mathcal{P}_k \gets \emptyset$ \Comment{Set of raw candidate curves propagated to $A_k$}
        
        \For{\textbf{each} cell $(i,j)$ such that $i + j = k$ \textbf{and} $1 \le i \le M-1, \, 1 \le j \le N-1$}
            \State $f_{\text{bottom}} \gets \text{active subsegments on } A_{k-1} \text{ entering bottom edge of } (i,j)$
            \State $f_{\text{left}} \gets \text{active subsegments on } A_{k-1} \text{ entering left edge of } (i,j)$
            \State $Type \gets \textproc{GetCellType}(P_i, Q_j)$
            
            \For{\textbf{each} piece $f \in \{f_{\text{bottom}} \cup f_{\text{left}}\}$}
                \State $\mathcal{P}_k \gets \mathcal{P}_k \cup \textproc{PropagatePiece}(f, Type, P_i, Q_j, \sigma)$
            \EndFor
        \EndFor
        
        \State $A_k \gets \textproc{GlobalLowerEnvelope}(\mathcal{P}_k)$ \Comment{Davenport--Schinzel Envelope}
    \EndFor
    
    \State \Return $A_{M+N-2}(p, q)$ \Comment{$\eps_0$-additive-approximate CDTW cost}
\EndProcedure
\end{algorithmic}
\end{algorithm}

\begin{algorithm}[H]
\caption{PropagatePiece}
\begin{algorithmic}[1]
\Procedure{PropagatePiece}{$f(s), Type, P_i, Q_j, \sigma$}
    \State $I \gets \text{domain interval of } f(s)$
    \If{$Type = \text{Type 1}$}
        \State $\mathcal{F} \gets \textproc{StraightPropagation}(f(s), \text{output boundary})$
        \State \Return $\bigcup_{g \in \mathcal{F}} \textproc{QuinticApproximate}(g, \mathrm{domain}(g), \tau_P, \sigma)$

    \ElsIf{$Type = \text{Type 2}$}
        \State $L_{Z(h)} \gets \text{maximum parameter length of diagonal valley } Z(h)$
        \State $g_{\text{entry}}(v) \gets \textproc{StraightPropagation}(f(s), \text{diagonal valley } Z(h))$
        \State $g_Z(v) \gets \textproc{CumulativeMin}(g_{\text{entry}}, v, L_{Z(h)})$
        \State $\mathcal{F} \gets \textproc{StraightPropagation}(g_Z(v), \text{output boundary})$
        \State \Return $\bigcup_{g \in \mathcal{F}} \textproc{QuinticApproximate}(g, \mathrm{domain}(g), \tau_P, \sigma)$

    \ElsIf{$Type = \text{Type 3}$}
        \State $\{v_1, \dots, v_k\} \gets \text{discrete valley points of } Z(h)$
        \State $\mathcal{P}_{\text{out}} \gets \emptyset$
        \For{\textbf{each} $v_x \in Z(h)$}
            \State $g_{\text{entry}}(v_x) \gets \min_{s \in I} \big(f(s) + C(s, v_x)\big)$
        \EndFor
        \State $g_Z(v_1) \gets g_{\text{entry}}(v_1)$

        \For{$x = 2$ \textbf{to} $k$}
            \State $g_Z(v_x) \gets \min\big(g_{\text{entry}}(v_x),\, g_Z(v_{x-1}) + C(v_{x-1}, v_x)\big)$
        \EndFor

        \For{$x = 1$ \textbf{to} $k$}
            \State $\mathcal{F} \gets \textproc{StraightPropagation}(g_Z(v_x), \text{output boundary})$
            \State $\mathcal{P}_{\text{out}} \gets \mathcal{P}_{\text{out}} \cup \bigcup_{g \in \mathcal{F}} \textproc{QuinticApproximate}(g, \mathrm{domain}(g), \tau_P, \sigma)$
        \EndFor

        \State \Return $\mathcal{P}_{\text{out}}$
    \EndIf
\EndProcedure
\end{algorithmic}
\end{algorithm}

\begin{algorithm}[H]
\caption{QuinticApproximate}
\begin{algorithmic}[1]
\Procedure{QuinticApproximate}{$g(t)$, $[t_a, t_b]$, $\tau_P$, $\sigma$}
    \State $\mathcal{R}(t) \gets \mathrm{Res}_s(P, P_s)$ \Comment{\Cref{dfn:regime}}
    \State Isolate roots of $\mathcal{R}(t)$ in $[t_a, t_b]$
    \State Partition $[t_a, t_b]$ into $O(1)$ Regime 1 and 2 sub-intervals
    \State $\mathcal{P}_{\text{out}} \gets \emptyset$
    \For{\textbf{each} sub-interval $[u_a, u_b]$ in the partition}
        \If{$[u_a,u_b]$ is Regime 1}
            \State Bisect $[u_a, u_b]$ to target length $L =  O\big(r\,\sigma^{1/6}/M_g^{1/6}\big)$ \Comment{\Cref{prop:regularerror}}
            \State Choose $6$ Chebyshev nodes $t_0, \dots, t_5$ on the resulting interval
            \For{$k = 0$ \textbf{to} $5$}
                \State Refine $s^*(t_k)$ and evaluate $\tilde{y}_k \gets g(t_k)$ to $\eta$ bits \Comment{\Cref{thm:refinement}}
            \EndFor
            \State $p \gets$ Lagrange interpolant of $\{(t_k, \tilde{y}_k)\}_{k=0}^5$ \Comment{\Cref{thm:coeffbitsize}}
            \State $\mathcal{P}_{\text{out}} \gets \mathcal{P}_{\text{out}} \cup \{p\}$
        \Else \Comment{Regime 2}
            \State Excise buffer of half-width $\ell_0 =  O(\sigma^2/|c'|^2)$ around the singularity
            \State $\mathcal{P}_{\text{out}} \gets \mathcal{P}_{\text{out}} \cup \{\text{constant } g(t_a) \text{ on the buffer}\}$
            \State $\mathcal{P}_{\text{out}} \gets \mathcal{P}_{\text{out}} \cup \textproc{QuinticApproximate}(g, [u_a,u_b] \setminus \text{buffer}, \tau_P, \sigma)$ \Comment{recurse back to Regime 1}
        \EndIf
    \EndFor
    \State \Return $\mathcal{P}_{\text{out}}$
\EndProcedure
\end{algorithmic}
\end{algorithm}

\begin{algorithm}[H]
\caption{StraightPropagation}
\begin{algorithmic}[1]
\Procedure{StraightPropagation}{$f_{\text{in}}(s)$, Target Boundary}
    \State $I \gets \text{domain interval } [s_{\mathrm{start}}, s_{\mathrm{end}}] \text{ of } f_{\text{in}}(s)$
    \State Define $G(s,t) := f_{\mathrm{in}}(s) + \sqrt{D(s,t)} H(s,t)$
    \State Candidates $\mathcal{R}(t) \gets \{s_{\mathrm{start}}, s_{\mathrm{end}}\}$
    \State Define $\frac{\partial}{\partial s} G(s, t) = 0 \iff P(s,t) = 0$
    \State Solve $P(s,t) = 0$ for real algebraic root branches $\{s_1(t), s_2(t), \dots, s_d(t)\}$

    \For{\textbf{each} root branch $s_m(t)$}
        \If{$s_m(t) \in I$}
            \State $\mathcal{R}(t) \gets \mathcal{R}(t) \cup \{s_m(t)\}$
        \EndIf
    \EndFor

    \State $\mathcal{F}_{\text{out}} \gets \emptyset$
    \For{\textbf{each} valid candidate branch $s^*(t) \in \mathcal{R}(t)$}
        \State $\mathcal{F}_{\text{out}} \gets \mathcal{F}_{\text{out}} \cup \{G(s^*(t), t)\}$
    \EndFor

    \State \Return $\mathcal{F}_{\text{out}}$
\EndProcedure
\end{algorithmic}
\end{algorithm}

\begin{algorithm}[H]
\caption{CumulativeMin}
\begin{algorithmic}[1]
\Procedure{CumulativeMin}{$g_1(v), v, v_{\text{end}}$}
    \State $v_{\text{crit}} \gets \{u \mid g_1'(u) = 0\}$
    \State $g_{\min} \gets g_1$

    \For{\textbf{each} $u \in v_{\text{crit}}$}
        \If{$g_1$ is strictly increasing immediately after $u$}
            \State Find $v^* \in (u, v_{\text{end}}]$ such that $g_1(v^*) = g_1(u)$
            \If{no such $v^*$ exists}
                \State $v^* \gets v_{\text{end}}$
            \EndIf
            \State Replace $g_{\min}$ on interval $[u, v^*]$ with horizontal line constant $g_1(u)$
        \EndIf
    \EndFor

    \State \Return $g_{\min}$
\EndProcedure
\end{algorithmic}
\end{algorithm}

\begin{algorithm}[H]
\caption{GlobalLowerEnvelope}
\begin{algorithmic}[1]
\Require $\mathcal{P}_k = \{G_1, G_2, \dots, G_{m_k}\}$ pre-sorted spatially on $A_k$.
\Procedure{GlobalLowerEnvelope}{$\mathcal{P}_k$}
    \State Initialize active curve stack $S \gets \text{empty stack}$
    
    \For{$i = 1$ \textbf{to} $m_k$}
        \While{$S.\text{size}() \ge 1$}
            \State $G_{\text{top}} \gets S.\text{top}()$
            \State Find entry point $t_{\text{intersect}}$ where $G_i(t) = G_{\text{top}}(t)$
            \If{$G_i(t) \le G_{\text{top}}(t)$ for all $t \ge t_{\text{intersect}}$}
                \If{$t_{\text{intersect}} \le \text{start domain of } G_{\text{top}}$}
                    \State $S.\text{pop}()$
                \Else
                    \State Truncate domain of $G_{\text{top}}$ at $t_{\text{intersect}}$
                    \State \textbf{break}
                \EndIf
            \Else
                \State \textbf{break}
            \EndIf
        \EndWhile
        \State $S.\text{push}(G_i)$
    \EndFor
    
    \State $\mathcal{E}_k \gets \text{piecewise function constructed from curves remaining in } S$
    \State \Return $\mathcal{E}_k$
\EndProcedure
\end{algorithmic}
\end{algorithm}

\end{document}